\documentclass[journal]{IEEEtran}

\usepackage[english]{babel}
\usepackage{ifpdf}

\usepackage{cite} 
\usepackage{url}
\usepackage{hyperref}

\ifCLASSINFOpdf
	\usepackage[pdftex]{graphicx}
	\graphicspath{{./figures/}}
\else
	\usepackage[dvips]{graphicx}
	\graphicspath{./figures/}
\fi
\usepackage{color}
\usepackage{pgf, tikz, pgfplots}
\usetikzlibrary{shapes, arrows, automata}
\usetikzlibrary{calc,hobby,decorations}

\usepackage[cmex10]{amsmath}
\usepackage{amsfonts, amssymb, amsthm}
\usepackage{mathrsfs}

\usepackage{algorithm} 
\usepackage{algorithmic} 
\usepackage{amsmath} 
\usepackage{xcolor}
\usepackage{enumerate}
\usepackage{multirow}
\usepackage{rotating}
\usepackage{subcaption}
\usepackage{needspace}

\input{mySymbol.sty}

\definecolor{penndarkestblue}{cmyk}{1,0.74,0,0.77}
\definecolor{penndarkerblue}{cmyk}{1,0.74,0,0.70}
\definecolor{pennblue}{cmyk}{0.99,0.66,0,0.57} 
\definecolor{pennlighterblue}{cmyk}{0.98,0.44,0,0.35}
\definecolor{pennlightestblue}{cmyk}{0.38,0.17,0,0.17} 

\definecolor{penndarkestred}{cmyk}{0,1,0.89,0.66}
\definecolor{penndarkerred}{cmyk}{0,1,0.88,0.55}
\definecolor{pennred}{cmyk}{0,1,0.83,0.42} 
\definecolor{pennlighterred}{cmyk}{0,1,0.6,0.24}
\definecolor{pennlightestred}{cmyk}{0,0.43,0.26,0.12} 

\definecolor{penndarkestgreen}{cmyk}{1,0,1,0.68}
\definecolor{penndarkergreen}{cmyk}{1,0,1,0.57}
\definecolor{penngreen}{cmyk}{1,0,1,0.44} 
\definecolor{pennlightergreen}{cmyk}{1,0,1,0.25}
\definecolor{pennlightestgreen}{cmyk}{0.43,0,0.43,0.13}

\definecolor{penndarkestorange}{cmyk}{0,0.65,1,0.49}
\definecolor{penndarkerorange}{cmyk}{0,0.65,1,0.33}
\definecolor{pennorange}{cmyk}{0,0.54,1,0.24} 
\definecolor{pennlighterorange}{cmyk}{0,0.32,1,0.13}
\definecolor{pennlightestorange}{cmyk}{0,0.15,0.46,0.06}
	
\definecolor{penndarkestpurple}{cmyk}{0,1,0.11,0.86}
\definecolor{penndarkerpurple}{cmyk}{0,1,0.13,0.82}
\definecolor{pennpurple}{cmyk}{0,1,0.11,0.71} 
\definecolor{pennlighterpurple}{cmyk}{0,1,0.05,0.46}
\definecolor{pennlightestpurple}{cmyk}{0,0.35,0.02,0.23}
	
\definecolor{pennyellow}{cmyk}{0,0.20,1,0.05} 
\definecolor{pennlightgray1}{cmyk}{0,0,0,0.05}
\definecolor{pennlightgray3}{cmyk}{0.01,0.01,0,0.18}
\definecolor{pennmediumgray1}{cmyk}{0.04,0.03,0,0.31}
\definecolor{pennmediumgray4}{cmyk}{0.08,0.06,0,0.54}
\definecolor{penndarkgray2}{cmyk}{0.09,0.07,0,0.71}
\definecolor{penndarkgray4}{cmyk}{0.1,0.1,0,0.92}

\def\SO3{\mathrm{SO(3)}}

\newtheorem{assumption}{\hspace{0pt}\bf Assumption}
\newtheorem{lemma}{\hspace{0pt}\bf Lemma}
\newtheorem{proposition}{\hspace{0pt}\bf Proposition}

\newtheorem{theorem}{\hspace{0pt}\bf Theorem}

\newtheorem{remark}{\hspace{0pt}\bf Remark}

\newtheorem{definition}{\hspace{0pt}\bf Definition}

\begin{document}

\title{Stochastic Graph Neural Networks}

\author{\IEEEauthorblockN{Zhan Gao$^{ \dagger}$, Elvin Isufi$^{\ddagger }$ and Alejandro Ribeiro$^{\dagger}$} 
\thanks{The work in this paper was supported by ARL DCIST CRA W911NF-17-2-0181 and NSF HDR TRIPODS 1934960. Preliminary results appear in ICASSP 2020 conference \cite{Zhan2020}. $^{ \dagger}$Dept. of Electrical and Systems Eng., University of Pennsylvania (gaozhan, aribeiro@seas.upenn.edu). $^{\ddagger}$Dept. of Intelligent Systems, Delft University of Technology (e.isufi-1@tudelft.nl).} 
}

\markboth{IEEE TRANSACTIONS ON SIGNAL PROCESSING (ACCEPTED)}%
{Stochastic Graph Neural Networks}

\maketitle 

\begin{abstract}

Graph neural networks (GNNs) model nonlinear representations in graph data with applications in distributed agent coordination, control, and planning among others. Current GNN architectures assume ideal scenarios and ignore link fluctuations that occur due to environment, human factors, or external attacks. In these situations, the GNN fails to address its distributed task if the topological randomness is not considered accordingly. To overcome this issue, we put forth the stochastic graph neural network (SGNN) model: a GNN where the distributed graph convolution module accounts for the random network changes. Since stochasticity brings in a new learning paradigm, we conduct a statistical analysis on the SGNN output variance to identify conditions the learned filters should satisfy for achieving robust transference to perturbed scenarios, ultimately revealing the explicit impact of random link losses. We further develop a stochastic gradient descent (SGD) based learning process for the SGNN and derive conditions on the learning rate under which this learning process converges to a stationary point. Numerical results corroborate our theoretical findings and compare the benefits of SGNN robust transference with a conventional GNN that ignores graph perturbations during learning.
\end{abstract}

\begin{IEEEkeywords}
Graph neural networks, graph filters, distributed learning
\end{IEEEkeywords}

\IEEEpeerreviewmaketitle


\section{Introduction} \label{sec:intro}

Graph neural networks (GNNs) are becoming the predominant tool to learn representations for network data \cite{Scarselli2008, Wu2019} with resounding success in rating prediction \cite{Ying2018, Wu20192}, distributed agent coordination \cite{Tolstaya2019, NerveNet2018}, and learning molecular fingerprints \cite{Duvenaud2015, Fout2017}. One key property of GNNs is their distributed implementation. The latter yields GNNs suitable candidates for distributed learning over networks, where each node can compute its output by only communicating with its immediate neighbors \cite{Shuman2018, Gavili2017, Liu2018}. Applications include distributed agent coordination \cite{Tolstaya2019}, smart grid failure detection \cite{Owerko2018}, and control systems \cite{Zou2013}.

While in recent years we have experienced the proposal of several GNN architectures, they can be cast under three main streams: i) the message passing neural networks \cite{gori2005new, scarselli2008graph, battaglia2016interaction, gilmer2017neural}; ii) the convolutional graph neural networks \cite{henaff2015deep, Defferrard2016, Kiph2017, xu2018powerful, Fernando2019}; iii) the graph attention networks \cite{velivckovic2017graph, lee2018graph, wu2019dual}. Message passing neural networks leverage the graph structure as the computation graph and combine arbitrary information across edges. The advantage of such architectures is that it allows linking the ability of GNNs to discriminate different graphs with standard graph isomorphism tests \cite{morris2019weisfeiler}. Convolutional graph neural networks are inspired by CNNs in structured Euclidean domains (e.g., time and image). They replace temporal or spatial filters with graph filters \cite{Bruna2013, Defferrard2016, Fernando20202}, which are tools that generalize the convolution operation to the irregular graph domain \cite{Sandry2013, ortega2018}. In this way, a vast knowledge of distributed graph signal processing can be adopted to analyze these architectures and make them distributable as well as scalable. Lastly, the graph attention networks consider also learning the weights of the graph edges from data and are more appropriate for situations where such weights are known with some uncertainty. Recently, it has been shown that graph attention mechanisms are convolutional GNNs of order one learning over graphs with edge features \cite{Isufi2020}. Since convolutional GNNs admit the distributed implementation and are recently shown robust to small deterministic perturbations, we build upon such architectures.

While seminal to establish distributed learning with GNNs, current works consider the underlying topology fixed and deterministic when executing the GNN distributively. However, in practical applications involving sensor, communication, road, and smart grid networks, the graph connectivity changes randomly over time \cite{structural2004, Jenelius2006, Gungor2010}. In a robot coordination network, for instance, the communication graph may change due to agent malfunctions or communication links that fall with a certain probability. These random topological changes lead to a random graph filter \cite{Isufi17}; hence, to a random GNN output. That is, the GNN will be implemented distributively over random time-varying topologies in the testing phase that are mismatched with the deterministic topology used in the training phase, therefore, degrading the performance. In this work, we hypothesize this mismatch between testing and training phases should be addressed by developing an architecture that is trained on stochastic graphs representing the practical setting in the distributed implementation.

Processing and learning with the topological stochasticity has been recently investigated in graph signal processing and graph neural network literature. On the GSP front, \cite{Isufi17} studied graph filters over random topologies to characterize the robustness and reduce the computational cost. Subsequently, \cite{saad2020quantization} extended such analysis to consider also the effect of data quantization in a distributed setting, while \cite{Isufi2019} proposed sparse controlling strategies for random networks. On the GNN front, \cite{monti2017geometric, berg2017graph} dropped randomly nodes and edges of the data graph to improve recommendation diversity in a matrix completion task and a similar setting has also been investigated in \cite{velivckovic2017graph} with the attention mechanism. Subsequently, \cite{feng2020graph} considered dropping nodes randomly as a strategy to perform data augmentation over graphs, while \cite{rong2019dropedge} proposed dropping edges as a regularization technique to prevent the over-smoothing for training deep convolutional GNNs of order one \cite{Kiph2017}. However, the aforementioned works impose randomness on the graph structure during training to solve a specific centralized task run over the deterministic graph during testing. Contrarily, we consider a different setting where the convolutional GNN of any order is run distributively over physical networks. The graph randomness arises naturally due to external factors, and has to be accounted during training to improve the transference ability of the GNN to random graph scenarios during testing. Furthermore, we do not limit ourselves to numerical evaluation but provide thorough theoretical analysis by leveraging concepts from graph signal processing \cite{Shuman2013, Gama2019}.

Specifically, we propose the stochastic graph neural network (SGNN) model to account for random link losses during training. The SGNN incorporates the random graph realizations into the architecture and makes each node rely on its neighbors with uncertainty. This neighborhood uncertainty makes the learned parameters robust to topological fluctuations encountered during testing; hence, it endows the SGNN with the ability to transfer more robustly to scenarios where the graph topology is perturbed randomly. To characterize this \emph{robust transference}, we analyze \emph{theoretically} the variance of the SGNN output and highlight the role played by the different actors such as the filter properties, the link sampling probability, and the architecture width and depth. More in detail, our three main contributions are:

\smallskip
\begin{enumerate}[(i)]

\item \emph{Stochastic graph neural networks (Section \ref{sec:StoGNN})}: We define the SGNN as a similar layered architecture to the GNN, but where stochastic graph filters are employed during learning. These filters account for the underlying topological variations in the node data exchanges to build random higher-level features in the graph convolutional layer.

\item \emph{Variance analysis (Section \ref{sec:expper})}: We prove the variance of the SGNN output is bounded by a factor that is quadratic in the link sampling probability. To conduct such analysis, we develop the concept of generalized filter frequency response, which allows for the spectral analysis over random time-varying graphs. We also put forth the generalized Lipschitz condition for stochastic graph filters, which generalizes the conventional Lipschitz condition used in the deterministic setting \cite{Fernando2019}. This variance analysis holds uniformly for all graphs, and indicates the effects the architecture depth and width as well as the nonlinearities have on the SGNN performance [Theorem \ref{theorem1}].

\item \emph{Convergence analysis (Section \ref{sec:convergence})}: We postulate the SGNN learning problem that accounts for the graph stochasticity in the cost function. We develop a stochastic gradient descent algorithm for this learning procedure and derive conditions under which a stationary point is achieved. We further prove the convergence rate is proportional to the inverse square root of the number of iterations [Theorem \ref{theorem2}].

\end{enumerate}

Numerical results on source localization and robot swarm control corroborate our model in Section \ref{Numerical}. The conclusions are drawn in Section \ref{sec_conclusions}. All proofs are collected in the appendix.



\def \RESGp {\text{RES}(\ccalG,p)}

%

\section{Random Edge Sampling Graph Model} \label{sec:StoGNN}

Consider symmetric unweighted graph $\ccalG= (\ccalV, \ccalE)$. The vertex set contains $N$ nodes $\ccalV = \{ 1,\ldots,N \}$ and the edge set contains $M$ undirected edges $(i,j)= (j,i)\in \ccalE$. The adjacency matrix $\bbA \in \reals^{N \times N}$ has entries $[\bbA]_{ij} > 0$ if node $i$ is connected to node $j$ (i.e. $(i,j)\in \ccalE$) and $[\bbA]_{ij} =0$ otherwise, and the Laplacian matrix is $\bbL=\diag(\bbA\bbone)-\bbA$. To keep discussion general, introduce the shift operator $\bbS \in \reals^{N \times N}$ as a stand in for either the adjacency or Laplacian matrix of $\ccalG$. Symmetry of the graph implies symmetry of the shift operator, $\bbS =\bbS^\top$. 

Our interest is the distributed processing architecture over physical networks where edges can drop randomly. For instance, in wireless sensor networks, communication links break randomly due to channel noise. Other applications include robot swarm coordination, smart grids and traffic networks. In these cases, edges drop independently due to external factors and this results in stochastic graph topologies that we model with the random edge sampling (RES) model \cite{Isufi17}.

\begin{definition}[Random Edge Sampling Graph]\label{def_res} For a given graph $\ccalG = (\ccalV, \ccalE)$ and edge inclusion probability $p$, we define $\RESGp$ a random graph with realizations $\ccalG_k = (\ccalV, \ccalE_k)$ such that edge $(i,j)$ is in $\ccalE_k$ with probability $p$,
\begin{equation}\label{eqn_def_res}
   \text{Pr} \Big[ (i,j) \in \ccalE_k \Big] = p, \quad 
        \text{for all \ } (i,j) \in \ccalE.
\end{equation}
Edge inclusions in $\ccalE_k$ are drawn independently. 
\end{definition}

As per Definition \ref{def_res}, realizations $\ccalG_k = (\ccalV, \ccalE_k)$ have edges $(i,j) \in \ccalE_k$ drawn from the edges of $\ccalG$ independently with probability $p$. The graph realization $\ccalG_k$ induces a realization of an adjacency matrix $\bbA_k$. If we let $\bbB_k$ be a symmetric matrix with independently drawn Bernoulli entries $b_{k,ij}=b_{k,ji}$ we can write the adjacency of graph $\ccalG_k$ as the Hadamard product 
\begin{equation}
   \bbA_k = \bbB_k \circ \bbA.
\end{equation}
The Laplacian of $\ccalG_k$ is $\bbL_k=\diag(\bbA_k\bbone)-\bbA_k$. We will use $\bbS_k$ to denote either. We emphasize that the choice of $\bbS$ and $\bbS_k$ are compatible. We either have $\bbS=\bbA$ and $\bbS_k=\bbA_k$ or $\bbS=\bbL$ and $\bbS_k=\bbL_k$. For future reference define the expected shift operator as $\barbS = \mbE[\bbS]$ and the expected graph as $\bar{\ccalG}$ -- the one with matrix representation $\barbS$. We remark that it is ready to have edge dropping probabilities depending on nodes or edges. We make them equal to simplify expressions.

\begin{figure*}[t]
\centering
\includegraphics [width = 0.205\linewidth]
                 {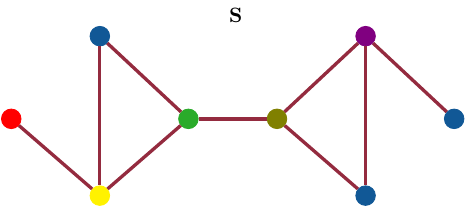}\qquad
\includegraphics [width = 0.205\linewidth]
                 {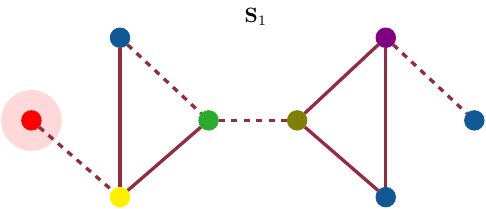}\qquad
\includegraphics [width = 0.205\linewidth]
                 {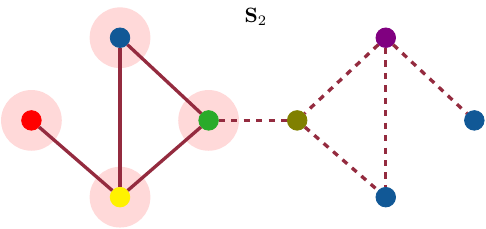}\qquad
\includegraphics [width = 0.205\linewidth]
                 {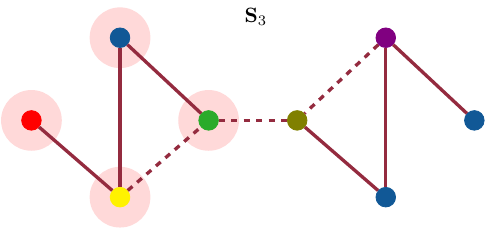} \\ \bigskip

\def \thisplotscale {1.8}
\def \unit {\thisplotscale cm}

\tikzstyle {Phi} = [rectangle,
                    thin,
                    minimum width = 1.0*\unit,
                    minimum height = \sumshift*\unit,
                    anchor = west,
                    draw,
                    fill = blue!20]

\tikzstyle {sum} = [circle,
                    thin,
                    minimum width  = 0.3*\unit,
                    minimum height = 0.3*\unit,
                    anchor = center,
                    draw,
                    fill = blue!20]

\def \deltax {1.9}
\def \deltay {0.8}
\def \sumshift {0.4}

\begin{tikzpicture}[x = 1*\unit, y = 1*\unit]

\node (first) [] {};

\path (first) ++ (0.15*\deltax, 0) node (0) [Phi] {$\bbS_0 = \bbI$};
\path (0)     ++ (0.9*\deltax, 0) node (1) [Phi] {$\bbS_1$};
\path (1)     ++ (1.0*\deltax, 0) node (2) [Phi] {$\bbS_2$};
\path (2)     ++ (1.0*\deltax, 0) node (3) [Phi] {$\bbS_3$};

\path (3.east) ++ (0.7*\sumshift*\deltax, 0) node [anchor=west] (last) [] {};

\path[-stealth] (first) edge [above            ] node {$\bbx$}               (0);	
\path[-stealth] (0)     edge [above, near start] node {$\ \bbS_0\bbx$}   (1);	
\path[-stealth] (1)     edge [above, near start] node {$\ \bbS_{1:0}\bbx$} (2);	
\path[-stealth] (2)     edge [above, near start] node {$\ \bbS_{2:0}\bbx$} (3);\path[-]        (3)     edge [above, near end  ] node {$\ \bbS_{3:0}\bbx$} (last);				
\path (0.east) ++ (\sumshift*\deltax, -\deltay) node (sum0) [sum] {$+$};
\path (1.east) ++ (\sumshift*\deltax, -\deltay) node (sum1) [sum] {$+$};
\path (2.east) ++ (\sumshift*\deltax, -\deltay) node (sum2) [sum] {$+$};
\path (3.east) ++ (0.7*\sumshift*\deltax, -\deltay) node (sum3) [sum] {$+$};

\path[-stealth, draw] (sum0 |- 0) --node[right] {$ h_0 $} (sum0); 	
\path[-stealth, draw] (sum1 |- 1) --node[right] {$ h_1 $} (sum1);	
\path[-stealth, draw] (sum2 |- 2) --node[right] {$ h_2 $} (sum2);	
\path[-stealth, draw] (sum3 |- 3) --node[right] {$ h_3 $} (sum3);	

\path[-stealth, draw] (sum0) -- (sum1);	
\path[-stealth, draw] (sum1) -- (sum2);	
\path[-stealth, draw] (sum2) -- (sum3);	

\path[-stealth] (sum3) edge [above] node
                {~~$\bbH(\bbS_{3:0})\bbx$} ++ (0.4*\deltax, 0);

\end{tikzpicture}
\caption{Stochastic Graph Filters. (Top-left) Underlying graph $\bbS$ and graph signal indicated by different colors. We focus on how the information is aggregated at the red node. The $\RESGp$ realization $\bbS_1$ drops the link between the red and yellow node (shown by dashed lines). In the shift $\bbS_{1:0}\bbx$, the red node uses only its own information (shown by the shaded area) as no information is obtained by its only neighbor (the yellow node). The $\RESGp$ graph realization $\bbS_2$ preserves the link red-yellow, therefore, from the shift $\bbS_{2:0}\bbx$ the red node gets information from its two hop neighbors (the blue and green node). This is because in $\bbS_1\bbx$ the yellow node aggregated information from its blue and green neighbors. At $\RESGp$ graph realization $\bbS_3$, the red node will not get information from the olive node when computing $\bbS_{3:0}$ because the link green-yellow was never active in all three $\RESGp$ realization $\bbS_1$, $\bbS_2$, and $\bbS_3$. Hence, the coverage of the red node remains limited to its two hop neighbors. (Bottom) Block diagram of the stochastic graph filter of order $K = 3$.} 
\label{fig.evRecMain}
\end{figure*}
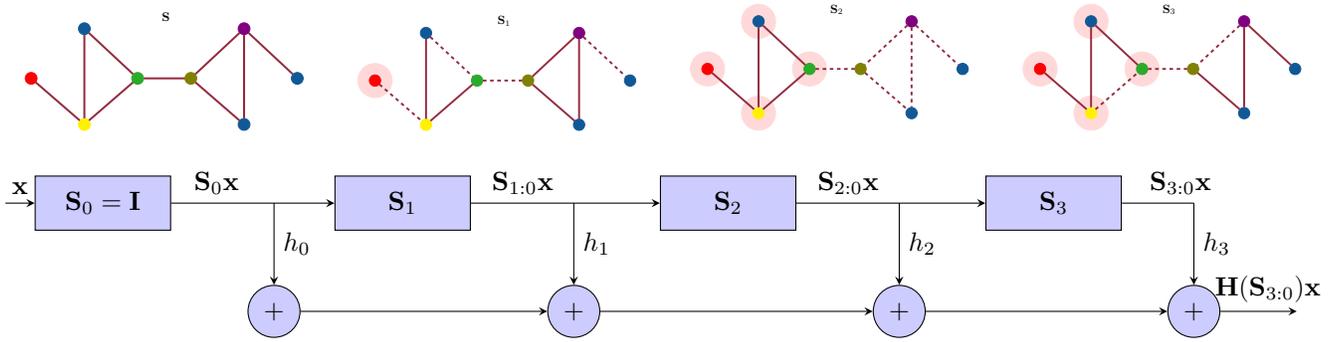

\subsection{Stochastic Graph Neural Network}

A stochastic graph neural network (SGNN) on $\ccalG$ is a graph neural network (GNN) run on a sequence of random realizations of a $\RESGp$ random graph. To be precise, let $\bbx = [x_1, \ldots, x_N]^\top \in \reals^N$ be a graph signal supported on $\ccalG$ in the sense that entry $x_i$ is associated to vertex $i$ \cite{Shuman2018, Gavili2017, Liu2018}. Further consider the sequence of shift operators $\bbS_k$ associated to graphs $\ccalG_k$ drawn independently as per Definition \ref{def_res}. The diffusion sequence is a corresponding collection of signals $\bbx^{(k)}$ expressed recursively as $\bbx^{(k)} := \bbS_k\bbx^{(k-1)}$ with $\bbx^{(0)}= \bbx$. Agreeing by convention that $\bbS_0 = \bbI$ and unrolling the recursion, the signals in the diffusion sequence are 
\begin{equation} \label{eq.sgf0}
   \bbx^{(k)} \ =  \ \bbS_{k}\bbx^{(k-1)}
              \ =  \ \big(\bbS_{k} \cdots \bbS_0\big)\, \bbx
              \ := \ \bbS_{k:0}\bbx
\end{equation}
where we have defined $\bbS_{k:0}:=\bbS_k \cdots \bbS_0$ in the last equality. The graph signal $\bbx^{(k-1)}$ is diffused over graph $\ccalG_k$ to produce the signal $\bbx^{(k)}$. This means the signal $\bbx^{(k)}$ is a diffused version of $\bbx=\bbx^{(0)}$ over the random sequence of graphs $\bbS_0,\ldots,\bbS_{k}$. The main motivation for this randomly time varying diffusion process is a communication network in which $\ccalE$ represents a set of possible links and $\ccalE_k$ a set of links that are activated at time index $k$. The model also applies to a social network where contacts are activated at random. See Section \ref{Numerical}.

We use the diffusion sequence in \eqref{eq.sgf0} to define a graph convolutional filter over the random $\RESGp$ graph. Fix a filter order $K$ and introduce $K+1$ coefficients $h_k$ with $k=0,\ldots,K$. The graph convolution of $\bbx$ with coefficients $h_k$ is a linear combination of the entries of the diffusion sequence in \eqref{eq.sgf0} modulated by coefficients $h_k$,
\begin{equation} \label{eq:sgf2}
   \bbu \ =  \ \sum_{k=0}^K h_{k}\bbx^{(k)} 
        \ =  \ \sum_{k=0}^K h_{k} \bbS_{k:0} \bbx 
        \ := \ \bbH (\bbS_{K:0}) \bbx
\end{equation}
where we defined the graph filter $\bbH(\bbS_{K:0}) := \sum_{k=0}^K h_{k} \bbS_{k:0}$ in which $\bbS_{K:0}=(\bbS_K, \ldots, \bbS_1, \bbS_0)$ represents the sequence of shift operators that appear in the filter\footnote{The notation $\bbS_{k:k'}=(\bbS_k, \ldots, \bbS_{k'})$ is a sequence of shift operators when it appears in an argument as in \eqref{eq:sgf2}. It represents the product $\bbS_{k:k' }:=\bbS_k \cdots \bbS_{k'}$ when it is a term in an expression as in \eqref{eq.sgf0}.}. This expression generalizes graph convolutional filters to settings where the topology changes between shifts \cite{Isufi17}. We shall refer to $\bbH(\bbS_{K:0})$ as a stochastic graph filter -- see Figure \ref{fig.evRecMain}. 

To build an SGNN relying on the graph filters in \eqref{eq:sgf2}, we consider the composition of a set of $L$ layers. The first layer $\ell=1$ consists of a bank of $F$ filters $\bbH^f_1(\bbS_{K:0})$ with coefficients $h_{1k}^f$ each of which produces the output graph signal $\bbu_{1}^{f}$. These filter outputs are passed through a pointwise nonlinear function $\sigma(\cdot)$ to produce a collection of $F$ features $\bbx_{1}^{f}$ that constitute the output of layer 1,
\begin{equation}\label{eq.sgnn_layer_1}
   \bbx_{1}^{f}
       \!=\! \sigma \!\left[\,\bbu_{1}^{f} \,\right] 
       \!=\! \sigma \!\left[\bbH^f_1(\bbS_{K:0})\bbx	\right]
       \!=\! \sigma \!\left[\,\sum_{k=0}^K h_{1k}^{f}\,\bbS_{1,k:0}^f\,\bbx \,\right]\!.
\end{equation}
The notation $\sigma[\bbu_{1}^{f}]$ signifies the vector $[\sigma(u_{11}^{f}),\ldots,\sigma(u_{1N}^{f})]^\top$ where the function $\sigma$ is applied separately to each entry of $\bbu_{1}^{f}$. We further emphasize that shift sequences $\bbS_{1,k:0}^f$ are specific to the feature $f$ and drawn independently from the random graph $\RESGp$. 

At subsequent intermediate layers  $\ell=2,\ldots,L-1$ the output features $\bbx_{\ell-1}^g$ of the previous layer, become inputs to a bank of $F^2$ filters with coefficients $h_{\ell k}^{fg}$ each of which produces the output graph signal $\bbu_{\ell}^{fg}$. To avoid exponential filter growth, the filter outputs derived from a common input feature $\bbx_{\ell-1}^g$ are summed and the result is passed through a pointwise nonlinear function $\sigma(\cdot)$ to produce a collection of $F$ features $\bbx_{\ell}^{f}$ that constitute the output of layer $\ell$
\begin{equation}\label{eq.sgnn}
   \bbx_{\ell}^{f}
       = \sigma \left[\, \sum_{g=1}^{F} \bbu_{\ell}^{fg} \, \right]
       = \sigma \left[\, \sum_{g=1}^{F} 
                            \sum_{k=0}^K 
                               h_{\ell k}^{fg} \,
                                  \bbS_{l,k:0}^{fg} \,
                                     \bbx_{\ell-1}^{g} \, \right]\!.
\end{equation}
The processing specified in \eqref{eq.sgnn} is repeated until the last layer $\ell=L$ in which we assume there is a single output feature which we declare to be the output of the GNN. To produce this feature, we process each input feature $\bbx_{L-1}^{g}$ with a graph filter with coefficients $h_{L k}^{g}$, sum all features, and pass the result through the pointwise nonlinear function $\sigma(\cdot)$. This yields the GNN output
\begin{equation}\label{eq.sgnn2}
   \bbPhi \big(\bbx;\bbS_{P:1}, \ccalH\big)
       \!\!=\! \sigma\! \left[\sum_{g=1}^{F} \bbu_{\ell}^{g}\!\right]
       \!\!=\! \sigma\! \left[\sum_{g=1}^{F} 
                            \sum_{k=0}^K 
                               h_{L k}^{g}
                                  \bbS_{L,k:0}^g
                                     \bbx_{L\!-\!1}^{g}\!\right]
\end{equation}
where the notation $\bbPhi(\bbx; \bbS_{P:1}, \ccalH)$ explicits that the SGNN output depends on the signal $\bbx$, a sequence of $P=K[2F+(L-2)F^2]$ independently chosen shift operators $\bbS_{P:1}$, and the filter tensor $\ccalH$ that groups filter coefficient $h_{\ell k}^{fg}$ for all layers $\ell$, orders $k$ and feature pairs $(f,g)$. The SGNN output in \eqref{eq.sgnn} is stochastic because the sequence of shift operators are realizations of the $\RESGp$ graph in \eqref{def_res}.

\begin{remark}
\normalfont The SGNN processes $F$ features with a bank of $F^2$ stochastic graph filters at each layer [cf. \eqref{eq.sgnn_layer_1}-\eqref{eq.sgnn2}]. In fact, it is equivalent to processing a $F$ dimensional vector-valued graph signal, i.e., the filter bank processing in \eqref{eq.sgnn} can be rewritten as
\begin{align}\label{eq:filterMatrix}
\bbX_\ell = \sigma \left[\sum_{k=0}^K \bbS_{k:0} \bbX_{\ell-1} \bbH_{\ell k}\right],~\forall~\ell=1,\ldots,L
\end{align}
where $\bbX_{\ell} = [\bbx_{\ell}^1,\ldots, \bbx_{\ell}^F]$ is the matrix collecting vector-valued graph signals at layer $\ell$ and $\bbH_{\ell k} \in \mathbb{R}^{F \times F}$ is the filter coefficient matrix. Without loss of generality, we can consider \eqref{eq:filterMatrix} as a multi-dimensional graph filter where the filter coefficients are matrices and the graph signal at each node is vector-valued. While the equivalence, we prefer the representation \eqref{eq.sgnn} because it is consistent with the conventional way the graph filter is defined \cite{Defferrard2016, Segarra2017} and illustrates explicitly the shift-sum principle behind the graph filter definition.
\end{remark}

\subsection{Filter Tensor Training}

To train the SGNN in \eqref{eq.sgnn} we are given a training set $\ccalT = \{(\bbx_r, \bby_r)\}$ made up of $R$ input-output pairs $(\bbx_r, \bby_r)$. We assume that both, input signals $\bbx_r$ and output signals $\bby_r$ are graph signals supported on $\bbS$. We are also given a cost function $c(\bby, \hby)$ to measure the cost of estimating output $\hby$ when the actual output is $\bby$. Our interest is on the cost of estimating outputs $\bby_r$ with the SGNN in \eqref{eq.sgnn} averaged over the training set
\begin{equation} \label{eq:PPP}
   C\big(\bbS_{P:1}, \ccalH\big) 
      = \frac{1}{R} \sum_{r=1}^R 
            c \Big(\bby_r, \bbPhi \big(\bbx_r;\bbS_{P:1}, \ccalH\big) \Big).
\end{equation}
Given that the SGNN output $\bbPhi(\bbx_r;\bbS_{P:1},\ccalH)$ is random, the cost $C\big(\bbS_{P:1}, \ccalH\big)$ is random as well. We therefore consider the cost averaged over realizations of $\RESGp$ and define the optimal filter tensor as the solution of the optimization problem
\begin{equation} \label{eq:perform}
   \ccalH^* = \argmin_{\ccalH} 
                   \mbE_{\bbS_{P:1}} 
                        \Big( C\big(\bbS_{P:1}, \ccalH\big) \Big).
\end{equation}
In the expectation in \eqref{eq:perform}, the shift operators $\bbS_k$ in the sequence $\bbS_{P:1}$ are drawn independently from the random $\RESGp$ graph of Definition \ref{def_res}. Training to optimize the cost in \eqref{eq:PPP} shall result in a SGNN with optimal \emph{average} cost. To characterize the cost \emph{distribution}, we study the variance of the SGNN output $\bbPhi \big(\bbx;\bbS_{P:1}, \ccalH\big)$ in Section \ref{sec:expper}. We also show that the SGNN can be trained with a stochastic gradient descent based learning process over $T$ iterations and prove this learning process converges to a stationary solution of \eqref{eq:perform} with a rate of $\ccalO(1/\sqrt{T})$ in Section~\ref{sec:convergence}. 

\begin{remark} \normalfont 
The SGNN defined by \eqref{eq.sgnn_layer_1}-\eqref{eq.sgnn2} admits a distributed implementation. This is because stochastic graph filters [cf. \eqref{eq:sgf2}] can be implemented in a distributed manner, which, in turn, is true because the diffusion sequence in \eqref{eq.sgf0} admits a distributable evaluation. To see the latter, denote the entries of $\bbS_{k}$ as $S_{k,ij}$ and the sparsity of $\bbS_{k}$ allows computing $\bbx^{(k)}$ as
\begin{equation} \label{eq:distributedImple}
  x^{(k)}_i = \sum_{j=1}^N S_{k,ij}\,x^{(k-1)}_j \!\!\!=\!\!\! \sum_{j:(j,i) \in\ccalE_k} S_{k,ij}\,x^{(k-1)}_j.
\end{equation}
Thus, the entry $x_i^{(k)}$ associated with node $i$ can be computed locally by only using the entries $x^{(k-1)}_j$ associated with nodes $j$ that are connected to node $i$ in the graph realization $\bbS_k$. Likewise, the $K$ entries $\{x_i^{(k)}\}_{k=1}^K$ can be computed through recursive information exchanges with neighboring nodes. Node $i$ can therefore compute the $i$th entry of the filter output as $u_i = \sum_{k=0}^K h_{k}x_i^{(k)}$ locally [cf. \eqref{eq:sgf2}]. Since the nonlinearity $\sigma(\cdot)$ is pointwise and thus, local, the SGNN inherits the distributed implementation. During deployment, each node need not know full knowledge of the graph, but only have communication capabilities to receive the neighborhood information and computational capabilities to aggregate the received information.
\end{remark}


\section{Variance analysis} \label{sec:expper}

The training optimizes the mean performance in \eqref{eq:perform}, while it says little about the deviation of a single realization around this mean. We quantify the latter in this section by providing an upper bound on the SGNN output variance. For our analysis, we consider the variance over all nodes
\begin{equation} \label{eq:var0}
\begin{split}
&{\rm var} \left[ \bbPhi(\bbx;\bbS_{P:1},\ccalH) \right] \!=\! \sum_{i=1}^N \!{\rm var}\left[[\bbPhi(\bbx;\bbS_{P:1},\ccalH)]_i\right]
\end{split}
\end{equation}
where $[\bbPhi(\bbx;\bbS_{P:1},\ccalH)]_i$ is the $i$th component of $\bbPhi(\bbx;\bbS_{P:1},\ccalH)$. This variance metric characterizes how individual entries $\{ [\bbPhi(\bbx;\bbS_{P:1},\ccalH)]_i \}_{i=1}^N$ deviate from their expectations, while it does not provide the correlation information. This is a typical way to characterize the average variance experienced at each individual entry in a multivariable stochastic system \cite{joshi2008sensor}. We conduct the variance analysis by pursuing in the graph spectral domain. The spectral analysis is recently used to characterize the stability of the GNN to small deterministic perturbations in the topology \cite{Fernando2019}. However, this model and its analysis method are inapplicable to the random edge sampling scenario because stochastic perturbations could result in large perturbation sizes. Therefore, we develop a novel approach to study the variance of the SGNN output by developing the novel concept of generalized filter frequency response, which allows for the spectral analysis over random graphs. Subsequently, we generalize the Lipschitz condition for graph filters to the stochastic setting to approach the variance analysis.

\subsection{Generalized Filter Frequency Response}

Consider the graph filter $\bbH(\bbS)$ [cf. \eqref{eq:sgf2} for $p=1$]. Since the shift operator $\bbS$ is symmetric, it accepts the eigendecomposition $\bbS = \bbV \bbLambda \bbV^{\top}$ with eigenvector basis $\bbV = [\bbv_{1}, \cdots, \bbv_{N}]$ and eigenvalues $\bbLambda = \text{diag} ([\lambda_{1},...,\lambda_{N}])$. We can expand signal $\bbx$ over $\bbV$ as $\bbx = \sum_{i=1}^N \hat{x}_i \bbv_{i}$; an operation known as the graph Fourier expansion of $\bbx$ \cite{ortega2018}. Vector $\hat{\bbx} = [\hat{x}_i, \cdots, \hat{x}_N]^\top$ contains the graph Fourier coefficients and it is called the graph Fourier transform (GFT) of $\bbx$. Substituting this expansion to the filter input $\bbu = \bbH(\bbS)\bbx$, we can write
 \begin{equation} \label{eq:mean11}
\begin{split}
\bbu \!=\! \sum_{k=0}^K h_k \bbS^k \sum_{i=1}^N \hat{x}_i \bbv_{i} \!=\! \sum_{i=1}^N \sum_{k=0}^K \hat{x}_i h_k \lambda_{i}^k \bbv_{i} .
\end{split}
\end{equation}
By further applying the Fourier expansion to the output $\bbu = \sum_{i=1}^N \hat{u}_i \bbv_i$, we get the input-output spectral filtering relation $\hat{\bbu} = \bbH(\bbLambda)\hat{\bbx}$. Here, $\bbH(\bbLambda)$ is a diagonal matrix containing the filter frequency response on the main diagonal. For filters in the form \eqref{eq:mean11}, the frequency response has the \emph{analytic} expression
 \begin{equation} \label{eq:frere1}
\begin{split}
h(\lambda) = \sum_{k=0}^K h_k \lambda^k.
\end{split}
\end{equation}
The graph topology instantiates the variable $\lambda$ to attain a value in the discrete set ${\lambda_1, \ldots, \lambda_N}$ and allows representing \eqref{eq:mean11} as a pointwise multiplication $\hat{u}_{i} = h(\lambda_{i})\hat{x}_i$ in the spectrum domain. The filter coefficients $\{h_k\}_{k=0}^K$ determine the shape of the frequency response function $h(\lambda)$. Figure \ref{fig:frequencyResponse} illustrates the latter concepts.

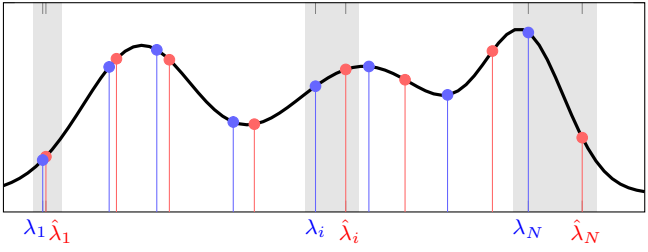
\begin{figure}[t]
\centering

\def \thisplotscale {3.48}
\def \unit {\thisplotscale cm}

\def \frequencyresponse 
     {   0.8*exp(-(1*(x-1.2))^2) 
       + 0.7*exp(-(0.7*(x-4))^2) 
       + 0.8*exp(-(1.4*(x-6))^2) 
       + 0.1}

{\footnotesize
\begin{tikzpicture}[x = 1.0*\unit, y=0.8*\unit]

\def \factorx {2.6/8}
\def \deltax  {0.5*\factorx}
\def \shadeshift  {0.05}

\path [fill=black!20, opacity = 0.5] 
              (\deltax - 0.001*\factorx - \shadeshift, 0.00) rectangle 
              (\deltax + 0.030*\factorx + \shadeshift, 1.00);
\path [fill=black!20, opacity = 0.5] 
              (\deltax + 3.193*\factorx - \shadeshift, 0.00) rectangle 
              (\deltax + 3.520*\factorx + \shadeshift, 1.00);
\path [fill=black!20, opacity = 0.5] 
              (\deltax + 5.648*\factorx - \shadeshift, 0.00) rectangle 
              (\deltax + 6.320*\factorx + \shadeshift, 1.00);

\begin{axis}[scale only axis,
             width  = 2.45*\unit,
             height = 0.8*\unit,
             xmin = -0.5, xmax=7.5,
             xtick = {0.03, -0.01, 3.77, 3.393, 6.72, 6.048},
             xticklabels = {\red{$\qquad\hat{\lambda}_1\phantom{\lambda}$},
                            \blue{$\lambda_1\ \ $}, 
                            \red{$\quad\hat{\lambda}_i\phantom{\lambda}$}, 
                            \blue{$\lambda_i$},
                            \red{$\quad\hat{\lambda}_{N}\phantom{\lambda}$},
                            \blue{$\lambda_N$}},
             ymin = -0, ymax = 1.15,
             ytick = {1.15},
             yticklabels = {},
             enlarge x limits=false]

\addplot+[samples at = {0.03, 0.91, 1.57, 
                        2.63, 3.77, 4.51, 
                        5.60, 6.72}, 
          color = red!60, 
          ycomb, 
          mark=otimes*, 
          mark options={red!60}]
         {\frequencyresponse};

\addplot+[samples at = {-0.01, 0.819, 1.413, 
                        2.367, 3.393, 4.059, 
                        5.04, 6.048}, 
          color = blue!60, 
          ycomb, 
          mark=oplus*, 
          mark options={blue!60}]
         {\frequencyresponse};

\addplot[ domain=-0.5:7.5, 
          samples = 80, 
          color = black,
          line width = 1.2]
         {\frequencyresponse};

\end{axis}
\end{tikzpicture}}

\caption{Frequency response of a graph filter (black line). The frequency response $h(\lambda)$ is an analytic function determined by the filter coefficients $\{ h_k \}_{k=0}^K$ and the function variable $\lambda$ depends on specific underlying graphs. We highlight the latter for two graphs $\ccalG$ and $\hat{\ccalG}$ with shift operators $\bbS$ and $\hat{\bbS}$. The former $\bbS$ instantiates $h(\lambda)$ on its eigenvalues $\lambda_1,\ldots, \lambda_N$ (in blue), while the latter $\hat{\bbS}$ instantiates $h(\lambda)$ on its eigenvalues $\hat{\lambda}_1,\ldots,\hat{\lambda}_N$ (in red). Changing graphs only instantiates $\lambda$ on different eigenvalues but does not change $h(\lambda)$ itself.}\label{fig:frequencyResponse}
\end{figure}

For the stochastic graph filter $\bbH(\bbS_{K:0})$ in \eqref{eq:sgf2}, we have a deterministic shift operator $\bbS_0 = \bbI_N$ and $K$ random shift operators $\bbS_1, \ldots, \bbS_K$. Since each $\bbS_k$ for $k = 1, \ldots, K$ is the shift operator of an undirected graph, it can be eigedecomposed as $\bbS_k = \bbV_k \bbLambda_k \bbV_k^\top$ with eigenvectors $\bbV_k = [\bbv_{k1}, \ldots, \bbv_{kN}]$ and eigenvalues $\bbLambda_k = \diag(\lambda_{k1}, \ldots, \lambda_{kN})$. We can now use these shift eigenvectors to compute a chain of graph Fourier decompositions each with respect to a different shift operator $\bbS_k$.

Starting from $\bbS_0 = \bbI_N$, we can write the Fourier expansion of signal $\bbx$ on the identity matrix as $\bbx = \bbS_0\bbx = \sum_{i_0 = 1}^N\hat{x}_{0i_0}\bbv_{0i_0}$, where $\bbv_{0i_0}$ is the $i_0$th column eigenvector of $\bbI_N$. When shifting the signal once over a RES($\ccalG, p$) graph, we have 
\begin{equation}\label{eq.dummy}
\bbx^{(1)} = \bbS_1\bbx = \sum_{i_0 = 1}^N \hat{x}_{0i_0}\bbS_1 \bbv_{0i_0}.
\end{equation}
We now treat each eigenvector $\bbv_{0i_0}$ as a new graph signal and compute its graph Fourier decomposition with respect to shift operator $\bbS_1 = \bbV_1\bbLambda_1\bbV_1^\top$. This, in turn, allows writing $\bbv_{0i_0}$ as $\bbv_{0i_0} = \sum_{i_1 = 1}^N\hat{x}_{1i_0 i_1}\bbv_{1i_1}$. Substituting the latter into \eqref{eq.dummy}, we have
\begin{equation}\label{eq.dummy1}
\bbx^{(1)} \!=\! \sum_{i_0 \!=\! 1}^N\hat{x}_{0i_0}\bbS_1\sum_{i_1 = 1}^N\hat{x}_{1i_0 i_1}\bbv_{1i_1} \!=\!\! \sum_{i_0 = 1}^N\sum_{i_1 = 1}^N\!\hat{x}_{0i_0} \hat{x}_{1i_0i_1}\lambda_{1i_1}\bbv_{1i_1}
\end{equation}
where $\{ \hat{x}_{0i_0} \}_{i_0=1}^N$ are $N$ Fourier coefficients of expanding signal $\bbx$ on $\bbS_0$, while $\{\hat{x}_{1i_0i_1}\}_{i_0i_1}$ are additional $N^2$ coefficients of extending this expansion to the two chain shift operator $\bbS_1\bbS_0$. Proceeding in a similar way, we can write the two-shifted signal $\bbx^{(2)} = \bbS_2\bbS_1\bbS_0\bbx$ as the Fourier decomposition on the three chain shift operator $\bbS_2\bbS_1\bbS_0$. In particular, exploiting $\bbx^{(2)} = \bbS_2\bbx^{(1)}$ and expansion \eqref{eq.dummy1}, we have
\begin{equation}\label{eq.dummy2}
\begin{split}
\bbx^{(2)} = \sum_{i_0 = 1}^N \sum_{i_1 = 1}^N \hat{x}_{0i_0} \hat{x}_{1i_0i_1} \lambda_{1i_1} \bbS_2 \bbv_{1i_1}.
\end{split}
\end{equation}
Therefore, treating again each eigenvector $\bbv_{1i_1}$ of $\bbS_1$ as a new graph signal and decomposing it on the Fourier basis of $\bbS_2 = \bbV_2\bbLambda_2\bbV_2^\top$, allows us to write \eqref{eq.dummy2} as
\begin{equation}\label{eq.dummy3}
\begin{split}
\bbx^{(2)} = \sum_{i_0 = 1}^N \sum_{i_1=1}^N \sum_{i_2=1}^N \hat{x}_{0i_0} \hat{x}_{1i_0i_1} \hat{x}_{2i_1i_2} \lambda_{1i_1} \lambda_{2i_2}\bbv_{2i_2}.
\end{split}
\end{equation}
In \eqref{eq.dummy3}, we have $\{\hat{x}_{2i_1i_2}\}_{i_1i_2}$ additional $N^2$ Fourier coefficients introduced by the three chain shift operator $\bbS_2\bbS_1\bbS_0$. Notice that while only the eigenvectors $\bbV_2$ of the last seen shift operator $\bbS_2$ are explicit in \eqref{eq.dummy3}, $\bbx^{(2)}$ is, nevertheless, influenced by all shift operators; especially, by their eigenspace alignment. The latter is captured by the GFT coefficients $\{\hat{\bbx}_{i_0}\}_{0i_0}$, $\{\hat{\bbx}_{1i_0 i_1}\}_{i_0 i_1}$, and $\{\hat{\bbx}_{2i_1i_2}\}_{i_1i_2}$.

Following this recursion, we can write the shifting of the graph signal $\bbx$ over $k$ consecutive RES($\ccalG, p$) realizations $\bbS_0, \ldots, \bbS_k$ as
  \begin{equation} \label{eq:gs}
\begin{split}
\bbx^{(k)}\!=\!\bbS_{k:0}\bbx \!=\!\! \sum_{i_0=1}^N\!\! \ldots\! \sum_{i_k=1}^N\! \hat{x}_{0i_0}\ldots \hat{x}_{ki_{k-1}i_k} \prod_{j=0}^k\lambda_{ji_j} \bbv_{ki_k}.
\end{split}
\end{equation}
Aggregating then the $K+1$ shifted signals $\bbx^{(0)}, \ldots, \bbx^{(K)}$, we can write the stochastic graph filter output $\bbu \!=\! \bbH(\bbS_{K:1})\bbx$ as
\begin{equation} \label{eq:gs2}
\begin{split}
\bbu \!=\! \sum_{i_0=1}^N \!\ldots\! \sum_{i_K=1}^N \sum_{k=0}^K \hat{x}_{0i_0}\ldots \hat{x}_{Ki_{K\!-\!1}i_K} h_k \prod_{j=0}^k\lambda_{ji_j} \bbv_{Ki_K}
\end{split}
\end{equation}
where $\{\hat{x}_{0i_0}\}$ and $\{ \hat{x}_{ki_ki_{k+1}} \}_{k=0}^{K-1}$ are Fourier coefficients of expanding $\bbx$ in the $(K+1)$ chain of shift operators $\bbS_{0}, \ldots, \bbS_{K}$ and $\{ \bbv_{Ki_K}\}_{i_K=1}^N$ are eigenvectors of the last seen shift operator $\bbS_K$. The output in \eqref{eq:gs2} is similar to \eqref{eq:mean11}. In fact, for $p = 1$ --\!\!\! when all shift operators are the same, deterministic, and their eigenspaces align perfectly-- \eqref{eq:gs2} reduces to \eqref{eq:mean11}. Therefore, we can consider \eqref{eq:gs2} as a graph filtering operation over a chain of different shift operators. Considering then a generic eigenvalue vector $\bblambda = [\lambda_{1},\ldots,\lambda_{K}]^\top$ (or graph frequency vector) in which $\lambda_k$ is the frequency variable for the shift operator $\bbS_k$, we define the $K$-dimensional \emph{analytic} generalized frequency response
\begin{equation} \label{eq:h2}
\begin{split}
h(\bblambda) = \sum_{k=0}^K h_k \lambda_{k} \cdots \lambda_{1} \lambda_0 = \sum_{k=0}^K h_k \lambda_{k:0}
\end{split}
\end{equation}
where $\lambda_{k:0} := \lambda_{k} \ldots \lambda_{1}\lambda_0$ is defined for convenience and $\lambda_0=1$ by definition (i.e., $\bbS_0 = \bbI_N$). The chain of shift operators $\bbS_1, \ldots, \bbS_K$ instantiates the generic vector $\bblambda$ to specific eigenvalues in each dimension to evaluate the analytic function $h(\bblambda)$, while the coefficients $\{ h_k \}_{k=0}^K$ determine the $K$-dimensional surface of the generalized frequency response $h(\bblambda)$. Figure \ref{fig:generalizedFrequencyResponse} shows an example of $h(\bblambda)$ in two dimensional space. 

Both the frequency response $h(\lambda)$ and the generalized frequency response $h(\bblambda)$ are characterized by the filter coefficients $\{ h_k \}_{k=0}^K$, while the specific shift operators only instantiate their variables. As such, by focusing directly on properties of $h(\lambda)$ and $h(\bblambda)$, we can analyze the filter performance and further the SGNN performance independently on specific shift operators (or graphs). 

\subsection{Variance of the Stochastic Graph Filter}

For mathematical tractability, we start by characterizing the variance of the stochastic graph filter. Consider the stochastic graph filter output $\bbu = \bbH(\bbS_{K:1})\bbx$ in \eqref{eq:sgf2}. The variance of $\bbu$ over all nodes is ${\rm var} \left[ \bbu \right] = \tr \left( \mathbb{E} \left[ \bbu \bbu^{\top} \right]- \bar{\bbu} \bar{\bbu}^{\top} \right)$ where $\bar{\bbu} = \mathbb{E}[\bbu] = \bbH(\bar{\bbS}) \bbx$ is the expected filter output and $\tr(\cdot)$ is the trace. Our goal is to upper bound the variance for any underlying shift operator $\bbS$ under the RES($\ccalG, p$) model. For this, we pursue the spectral domain analysis to be independent on specific eigenvalues and need the following assumptions.

\begin{figure}[t]
\centering
\includegraphics[width=0.95\linewidth , height=0.65\linewidth, trim=10 10 10 10]{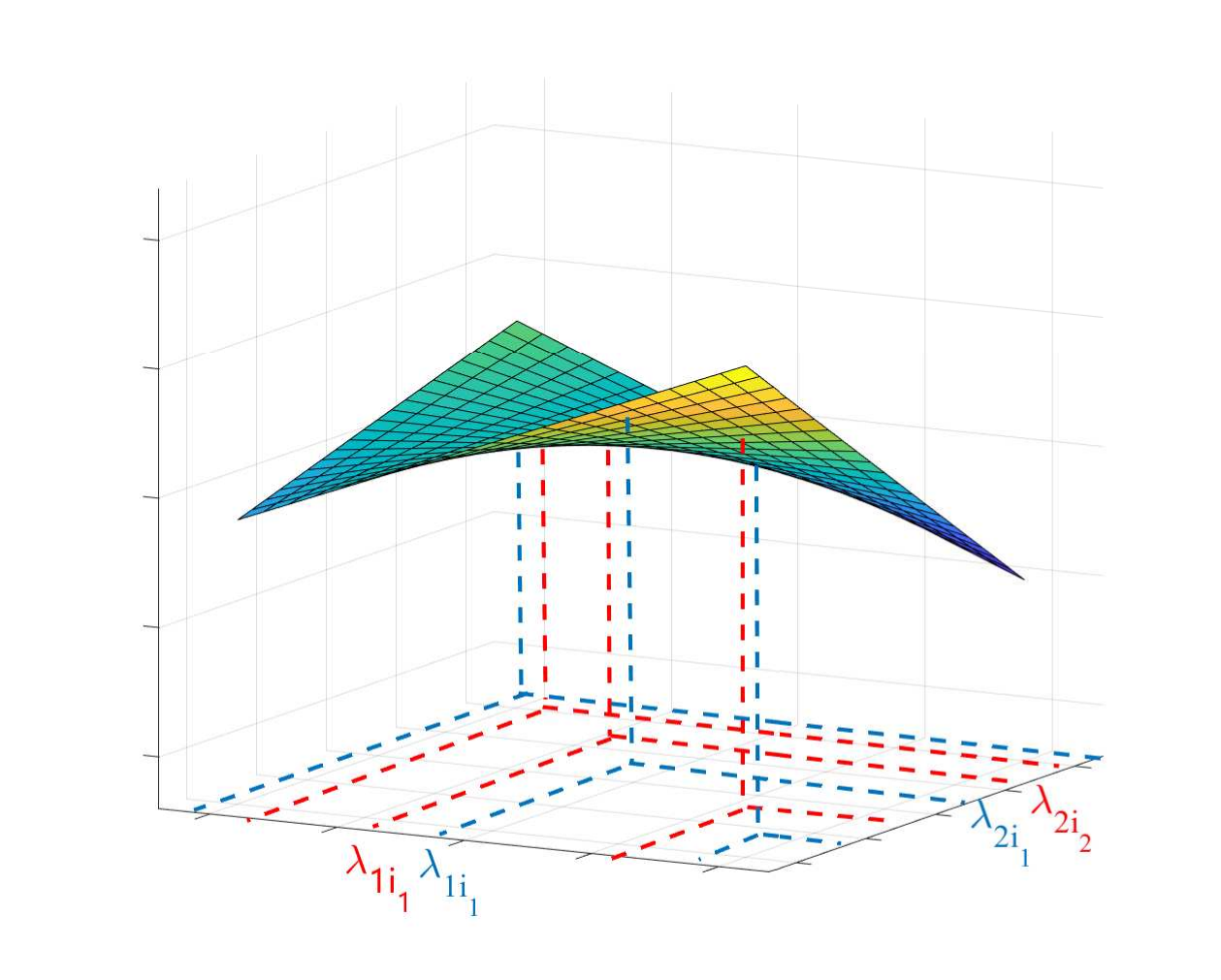}
\caption{The $2$-dimensional generalized frequency response of a stochastic graph filter. The transfer function $h(\bblambda)$ is independent of random graph realizations. For two specific graph realization chains, $h(\bblambda)$ is instantiated on specific eigenvalues determined by the chain of shift operators. We highlight the latter for two different graph realization chains with eigenvalue vectors in red and blue. Depending on the specific topology chain, the filter will have a different behavior. }\label{fig:generalizedFrequencyResponse}
\end{figure}

\begin{assumption} \label{as:3}
For a set of filter coefficients $\{ h_k \}_{k=0}^K$ and graph frequencies $\lambda$ in a finite set $\Lambda$, the filter frequency response $h(\lambda)$ in \eqref{eq:frere1} is bounded. I.e., there exists a constant $C_U$ such that for all $\lambda \in \Lambda$ the frequency response satisfies
\begin{equation} \label{eq:as3}
\begin{split}
\left|h(\lambda)\right| \le C_U.
\end{split}
\end{equation}
\end{assumption}

\begin{assumption} \label{as:4}
For a set of filter coefficients $\{ h_k \}_{k=0}^K$ and graph frequency vectors $\bblambda$ in a finite space $\Lambda^K$, the generalized filter frequency response $h(\bblambda)$ in \eqref{eq:h2} is Lipschitz. I.e., there exists a constant $C_g$ such that for any $\bblambda_1, \bblambda_2 \in \Lambda^K$ the generalized frequency response satisfies
\begin{equation} \label{eq:as4}
\begin{split}
\left|h(\bblambda_1) - h(\bblambda_2)\right| \le C_g \|\bblambda_1 - \bblambda_2\|_2.
\end{split}
\end{equation}
\end{assumption}

Assumption \ref{as:3} is commonly used in graph signal processing and states the filter frequency response in \eqref{eq:frere1} evolves within certain finite margins. Assumption \ref{as:4} indicates the stochastic graph filters are Lipschitz likewise the deterministic graph filters in \cite{Gama2019 } which are the particular case for $p = 1$. In other words, Assumption \ref{as:4} implies the generalized frequency response $h(\bblambda)$ does not change faster than linear in any frequency direction of $\bblambda$. With above preliminaries in place, the following proposition states the variance of the stochastic graph filter is upper bounded by a factor that is quadratic in the link sampling probability.

\begin{proposition} \label{proposition1}
Consider the stochastic graph filter $\bbH(\bbS_{K:0})$ of order K in \eqref{eq:sgf2} over a RES($\ccalG, p$) graph model [cf. Def. \ref{def_res}] of $M$ edges. Let $\bbS$ be the underlying shift operator, $\barbS$ the expected shift operator, and $\{h_k\}_{k=0}^K$ the filter coefficients. Let also the generalized frequency response $h(\bblambda)$ in \eqref{eq:h2} satisfy Assumption \ref{as:4} with constant $C_g$. 

Then, for any input signal $\bbx$, the variance of filter output $\bbH(\bbS_{K:0})\bbx$ is upper bounded as 
\begin{equation} \label{eq:propo3}
\begin{split}
{\rm var} \left[ \bbH(\bbS_{K:0})\bbx \right] \le  p(1-p) C \| \bbx \|_2^2+ \mathcal{O}(p^2(1-p)^2)
\end{split}
\end{equation}
with constant $C= 2\alpha M K C_g^2$ and scalar $\alpha$ that depends on the shift operator\! \footnote{For example, $\alpha=1$ if $\bbS$ is the adjacency matrix or $\alpha=2$ if $\bbS$ is the graph Laplacian.}.
\end{proposition}
\begin{proof}
See Appendix \ref{pr:proposition1}.
\end{proof}

Proposition \ref{proposition1} shows the variance of the filter output does not diverge and varies within a finite range that depends on the link sampling probability $p$. The bound represents the influence of random link fluctuations on the filter variance. When links are stable ($p \to 1$) or links are highly unstable ($p \to 0$), the variance is small indicating the filter output varies close to the expected value. For the extreme cases $p=0$ or $p=1$, the bound reduces to zero because all RES($\ccalG,0$) or RES($\ccalG,1$) realizations are deterministic graphs. The maximum variance is achieved for $p=0.5$, corresponding to the most uncertain case about the presence of links. Constant $C$ represents the role of the graph and filter: $M$ is the number of edges; $K$ and $C_g$ are the filter order and the Lipschitz constant of generalized frequency response, respectively.

\subsection{Variance of the SGNN}

The variance of the SGNN is directly influenced by the variance of the stochastic graph filter with additional effects of the nonlinearity and the layered architecture. Before claiming the main result, we require an assumption for the nonlinearity.

\begin{assumption} \label{assumption3}
The nonlinearity $\sigma(\cdot)$ satisfying $\sigma(0)\!=\!0$ is Lipschitz and variance non-increasing. I.e., there exists a constant $C_\sigma$ such that for all $x, y \!\in\! \mathbb{R}$ the nonlinearity satisfies
\begin{equation}\label{eq:assumption3}
\begin{split}
|\sigma(x)-\sigma(y)| \le C_\sigma |x-y|, ~ {\rm var}[\sigma(x)] \le {\rm var}[x].
\end{split}
\end{equation}
\end{assumption}

The Lipschitz nonlinearity is commonly used, which includes ReLU, absolute value and hyperbolic tangent. We show the variance non-increasing is mild by the following lemma that proves this property holds for ReLU and absolute value.

\begin{lemma} \label{lemma3}
Consider the nonlinearity is ReLU $\sigma(x)=\max(0,x)$ or absolute value $\sigma(x) = |x|$. For a random variable $x$ with any distribution, it holds that
\begin{gather}
\label{eq:lemma3} {\rm var} \left[ \sigma (x) \right] \le {\rm var} \left[ x \right].
\end{gather}
\end{lemma}

The proof of Lemma \ref{lemma3} is in the supplementary material. The following theorem then formally quantifies the variance of the SGNN output and details the role of the SGNN architecture on the variance.

\begin{theorem} \label{theorem1}
Consider the SGNN in \eqref{eq.sgnn} of $L$ layers and $F$ features, over a RES($\ccalG,p$) graph model [cf. Def. \ref{def_res}] of $M$ edges. Let $\bbS$ be the underlying shift operator, $\barbS$ the expected shift operator, and $\ccalH$ the SGNN filter tensor. Let the stochastic graph filters be of order $K$ with frequency response \eqref{eq:frere1} and generalized frequency response \eqref{eq:h2} satisfying Assumption \ref{as:3} with constant $C_U$ and Assumption \ref{as:4} with constant $C_g$. Let also the nonlinearity $\sigma(\cdot)$ satisfy Assumption \ref{assumption3} with constant $C_\sigma$. 

Then, for any input graph signal $\bbx$, the variance of the SGNN output is upper bounded as
\begin{equation}\label{eq:theorem1main}
\begin{split}
\!{\rm var}\! \left[ \bbPhi(\bbx;\bbS_{P:1},\ccalH) \right] \!\le\! p(1\!-\!p) C \| \bbx \|_2^2 \!+\! \mathcal{O}(p^2(1\!-\!p)^2)
\end{split}
\end{equation}
with constant $C= 2 \alpha M \sum_{\ell=1}^L \!F^{2L-3} C_\sigma^{2\ell-2} C_U^{2L-2} K C_g^2$ and scalar $\alpha$ that depends on the shift operator.
\end{theorem}
\begin{proof}
See Appendix \ref{pr:theorem1}.
\end{proof}

As it follows from \eqref{eq:propo3}, the SGNN variance bound has a similar form as the stochastic graph filter variance bound [cf. Proposition \ref{proposition1}], therefore, the conclusions of Proposition \ref{proposition1} apply also here. However, there is a large difference between the two bounds within constant $C$. In the SGNN, this constant is composed of three terms representing respectively the effect of the \emph{graph}, \emph{filter}, and \emph{neural network} on the variance. The graph impact is captured by $\alpha M$ which shows the role of shift operator type and that more connected graphs lead to a worse variance bound. The filter impact is captured by $C_U^{2L-2}KC_g^2$ which is dictated by the filter response $h(\lambda)$ and the generalized response $h(\bblambda)$. Notice that while we might be tempted to consider filters with a small $C_U$ [cf. Asm.\ref{as:3}] to have a smaller variance, it will lead to an information loss throughout the layers; hence to a lower performance. Our expectation is that filters with $C_U$ close to one will be a good tradeoff. The architecture impact is captured by $\sum_{\ell=1}^L \!F^{2L-3} C_\sigma^{2\ell-2}$, which is a consequence of the signal propagating throughout all $L$ layers and nonlinearities. The latter implies a deeper and wider SGNN has more uncertainty with a worse variance. This behavior is intuitively expected since the SGNN will have more random variable components and the aggregation of random variables in \eqref{eq.sgnn} will lead to a higher variance. The aforementioned factors are also our handle to design distributed SGNN architectures that are more robust to link losses.


\section{Convergence Analysis} \label{sec:convergence}

In this section, we propose an explicit learning process to train the SGNN with stochasticity appropriately. This learning process consists of minimizing the cost in \eqref{eq:PPP} by accounting for different RES($\ccalG, p$) realizations. The associated convergence analysis translates into identifying whether this learning process with stochasticity converges to a stationary point and if so under which conditions.

\subsection{Learning Process}\label{sec:learProc}

Consider the SGNN has a fixed sequence of RES($\ccalG, p$) shift operator realizations $\bbS_{P:1}$ in \eqref{eq.sgnn} when processing the input. We shall refer to the latter fixed architecture as an SGNN realization. The learning process of the SGNN follows that of the conventional GNN but now with stochastic graph filter coefficients $\ccalH$ learned through descent algorithms. The tensor $\ccalH$ is updated iteratively and each iteration $t$ comprises a forward and a backward phase. In the \emph{forward phase}, the SGNN has the tensor $\ccalH_t$ and a fixed architecture realization (i.e., the shift operator realizations $\bbS_{P:1}$ are fixed), processes all input signals $\{ \bbx_r \}_{r=1}^R$, obtains the respective outputs $\{ \bbPhi(\bbx_r; \bbS_{P:1}, \ccalH_t) \}_{r=1}^R$, and computes the cost $C(\bbS_{P:1},\ccalH_t)$ as per \eqref{eq:PPP}. In the \emph{backward phase}, the tensor $\ccalH_t$ gets updated with a gradient descent algorithm with step-size $\alpha_t$. This learning procedure incorporates the graph stochasticity in each gradient descent iteration $t$ by fixing an SGNN realization. The latter mimics the network randomness caused by practical link losses; hence, it makes the cost function $C(\bbS_{P:1},\ccalH)$ at each iteration random. 

While the cost function $C(\bbS_{P:1},\ccalH)$ is random at each iteration, it is sampled from the distribution $m_p(\bbS_{P:1})$ determined by the link sampling probability $p$. This motivates to focus on the deterministic expected cost over the distribution $m_p(\bbS_{P:1})$ instead of random single cost. From this intuition, we consider the following stochastic optimization problem
\begin{equation} \label{eq:mini}
\begin{split}
\min_{\ccalH} \bar{C}(\ccalH) = \min_{\ccalH} \mathbb{E} \left[C(\bbS_{P:1},\ccalH)\right].
\end{split}
\end{equation}
The problem in \eqref{eq:mini} is akin to the conventional stochastic optimization problem, while the expectation $\mathbb{E}[\cdot]$ is now w.r.t. the graph stochasticity instead of the data distribution. The random cost function $C(\bbS_{P:1},\ccalH_t)$ is entirely determined by the SGNN realization $\bbPhi(\cdot; \bbS_{P:1}, \ccalH_t)$. In turn, this indicates that fixing an SGNN realization is on the identical notion of randomly sampling a cost function $C(\bbS_{P:1},\ccalH_t)$ in \eqref{eq:mini}. The latter implies that the proposed learning process is equivalent to performing the stochastic gradient descent (SGD) method on the stochastic optimization problem in \eqref{eq:mini}. To be more precise, the \emph{forward phase} is equivalent to selecting a random cost function $C(\bbS_{P:1},\ccalH_t)$ from the stochastic distribution, and the \emph{backward phase} is equivalent to approximating the true gradient $\nabla_\ccalH \bar{C}(\ccalH)$ with $\nabla_\ccalH C(\bbS_{P:1},\ccalH)$ stochastically and uses this approximation to update the parameters in $\ccalH$ at each iteration. Put simply, we can interpret the learning process of the SGNN as running the SGD method on the stochastic optimization problem \eqref{eq:mini} -- see Algorithm  \ref{algo2}.

The graph stochasticity incorporated during training matches the graph stochasticity encountered during testing. The tensor will then be learned in a way that is robust to stochastic perturbations since each node does not rely certainly on information from all its neighbors. This robustness yields an improved transference ability, i.e., the SGNN can transfer better to scenarios where the graph topology changes randomly during testing. While the SGNN is expected to exhibit better performance in distributed tasks, the involved graph stochasticity makes the proposed learning process random. As such, it is unclear if the SGNN learning process converges. We analyze this aspect next and prove that under conventional mild conditions, the SGNN learning process reaches a stationary point.

\subsection{Convergence of SGNN Learning Process} \label{sec:conf}


{\linespread{0.9}
\begin{algorithm}[t]  \begin{algorithmic}[1]
\STATE \textbf{Input:} training set $\mathcal{T}$, cost $C(\bbS_{P:1},\ccalH)$, and initial filter tensor $\ccalH_0$
\FOR {$t = 0,1,...,T$}
      \STATE Fix a random cost function realization $C(\bbS_{P:1},\ccalH_t)$
      \STATE Compute the gradients $\nabla_\ccalH C(\bbS_{P:1},\ccalH_t)$
      \STATE Update the tensor with step-size $\alpha_t$\\
      $\ccalH_{t+1} = \ccalH_t - \alpha_t \nabla_\ccalH C(\bbS_{P:1},\ccalH_t)$ \\
\ENDFOR
\end{algorithmic}
\caption{Stochastic Gradient Descent on \eqref{eq:mini}}\label{algo2} \end{algorithm}}

Given the equivalence between the SGNN learning process and the SGD method in Algorithm \ref{algo2}, we analyze the convergence of the SGNN learning process by proving the convergence of the SGD counterpart; i.e., we prove there exists a sequence of tensors $\{\ccalH_t\}$ generated by the SGD that approaches a stationary point $\ccalH^*_s$ of \eqref{eq:mini}. Since problem \eqref{eq:mini} is nonconvex due to the nonlinear nature of the SGNN, we can no longer use the metric $\bar{C}(\ccalH) - \bar{C}(\ccalH^*_s)$ or $\|\ccalH - \ccalH_s^*\|_2^2$ as a convergence criterion. We instead use the gradient norm $\| \nabla_\ccalH \bar{C}(\ccalH)\|^2_2$, which is a typical surrogate to quantify stationarity and has also a similar order of magnitude as the above two quantities \cite{Ghadimi2013}. To render this analysis tractable, we assume the following.

\begin{assumption} \label{as:1}
The expected cost $\bar{C}(\ccalH)$ in \eqref{eq:mini} is Lipschitz continuous. I.e., there exists a constant $C_L$ such that
\begin{equation} \label{eq:as1}
\begin{split}
\!\!\!\| \nabla_\ccalH \bar{C}(\ccalH_1) \!-\!\! \nabla_\ccalH \bar{C}(\ccalH_2) \|_2 \!\le\! C_L\| \ccalH_1\!-\!\ccalH_2 \|_2
\end{split}
\end{equation}
for any $\ccalH_1$ and $\ccalH_2$.
\end{assumption}

\begin{assumption}  \label{as:2}
The stochastic cost $C(\bbS_{P:1},\ccalH)$ in \eqref{eq:PPP} has a $C_B$-bounded gradient norm. I.e., there exists a constant $C_B$ such that
\begin{equation}
\begin{split}
\| \nabla_\ccalH C(\bbS_{P:1},\ccalH) \|_2 \le C_B.
\end{split}
\end{equation}
\end{assumption}

Both assumptions are mild and common in optimization analysis \cite{Goldstein1977,Ghadimi2013}. They allow us to handle the stochasticity in the gradient and prove the convergence of the SGNN learning process. This is formalized by the following theorem.

\begin{theorem} \label{theorem2}
Consider the SGNN in \eqref{eq.sgnn} over a RES($\ccalG, p$) graph model [cf. Def. \ref{def_res}] with underlying shift operator $\bbS$ and the training set $\ccalT$. Consider the SGNN learning process for $T$ iterations, i.e., running the stochastic gradient descent on \eqref{eq:mini} for $T$ iterations [cf. Algo. \ref{algo2}]. Further, let $\ccalH^*$ be the global optimal solution of the cost function $\bar{C}(\ccalH)$ in \eqref{eq:mini} and let Assumptions \ref{as:1} and \ref{as:2} hold with respective constants $C_L$ and $C_B$. For any initial tensor $\ccalH_0$ and gradient step-size
\begin{equation} \label{eq:step}
\begin{split}
\alpha_t = \alpha = \sqrt{\frac{2\left( \bar{C}(\ccalH_0)-\bar{C}(\ccalH^*)\right)}{T C_L C_B^2 }}
\end{split}
\end{equation}
the minimum expected gradient square norm is bounded as
\begin{equation}\label{eq:thm22}
\begin{split}
\min_{0\le t \le T-1} \mathbb{E} \left[ \| \nabla_\ccalH \bar{C}(\ccalH_t) \|^2_2 \right] \le \frac{C}{\sqrt{T}}
\end{split}
\end{equation}
with constant $C\!=\! \sqrt{2\left( \bar{C}(\ccalH_0)-\bar{C}(\ccalH^*)\right)\!C_L}C_B$. That is, the learning process of the SGNN trained to minimize \eqref{eq:mini} converges to a stationary point with a rate of $\ccalO(1/\sqrt{T})$.
\end{theorem}

\begin{proof}
See Appendix \ref{pr:theorem2}.
\end{proof}

Theorem \ref{theorem2} states that with an appropriate choice of step-size $\alpha_t$, Algorithm 1 converges to a stationary point for tensor $\ccalH_t$. The step-size $\alpha_t$ in \eqref{eq:step} depends on the total number of iterations $T$; a more practical way is to set $\alpha_t \propto 1/t$ or $1/\sqrt{t}$. Due to the non-convexity of the SGNN, the learning process has guaranteed convergence to a local stationary minima. The latter can be extended to a better (potentially global) minimum with standard approaches such as training the SGNN multiple times.

\begin{figure}[t]
\centering
\includegraphics[width=0.45\textwidth, trim=10 10 10 10]{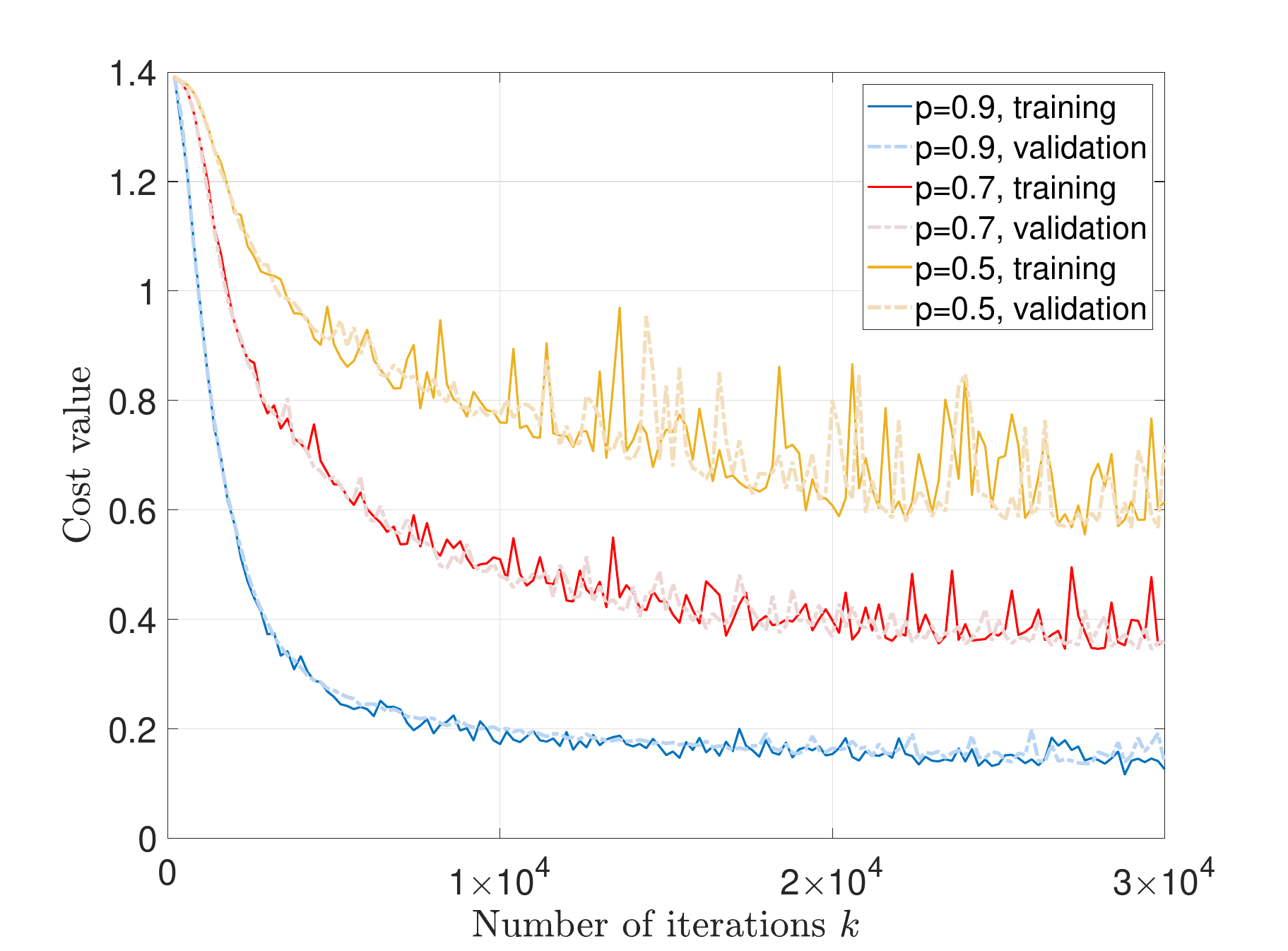}
\caption{Convergence of stochastic graph neural network with link sampling probabilities $p=0.9$, $p=0.7$ and $p=0.5$.}\label{fig:convergence}
\end{figure}

\begin{figure*}%
\centering
\begin{subfigure}{0.61\columnwidth}
\includegraphics[width=1.1\linewidth, height = 0.85\linewidth]{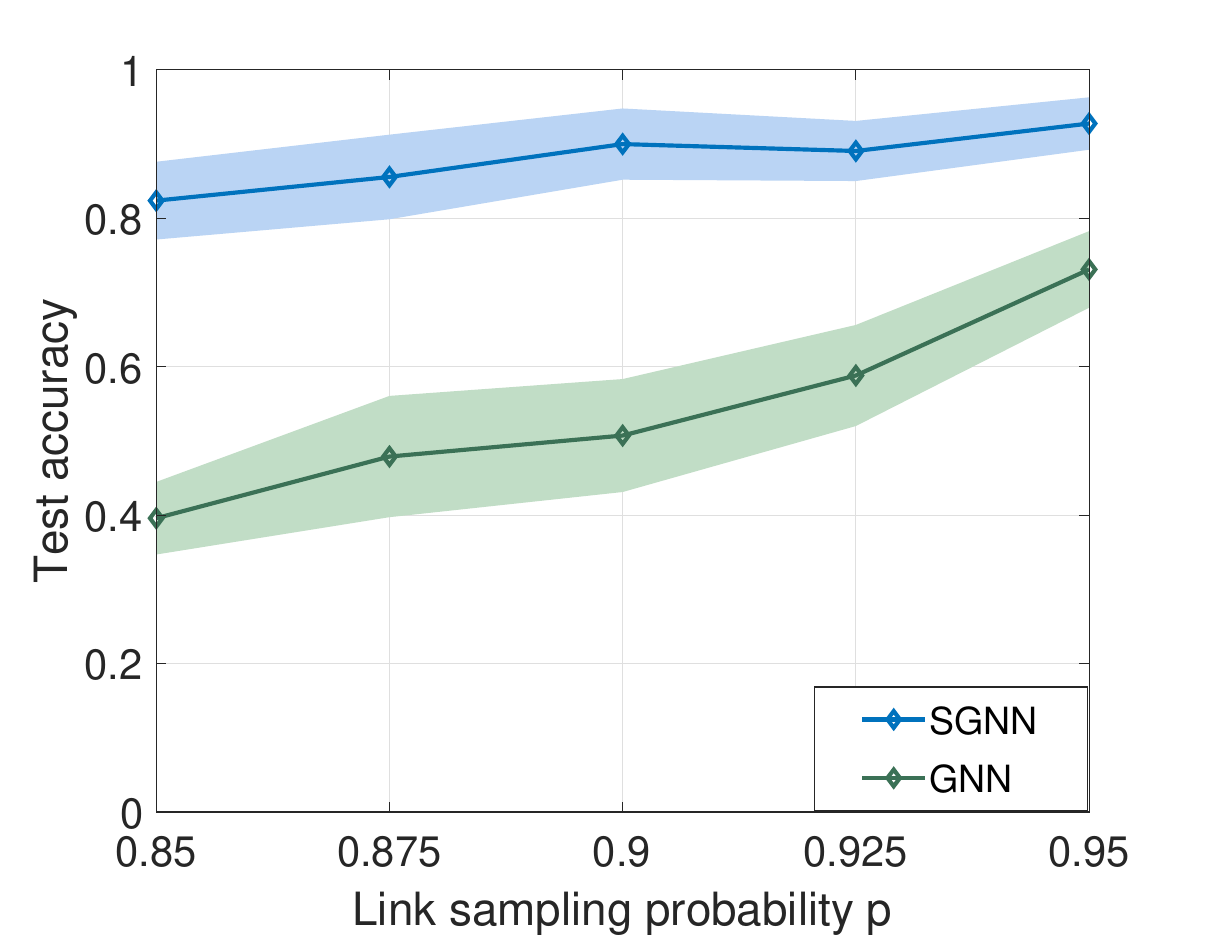}%
\caption{}%
\label{subfig:source_vary_p_zoom}%
\end{subfigure}\hfill\hfill%
\begin{subfigure}{0.61\columnwidth}
\includegraphics[width=1.1\linewidth,height = 0.85\linewidth]{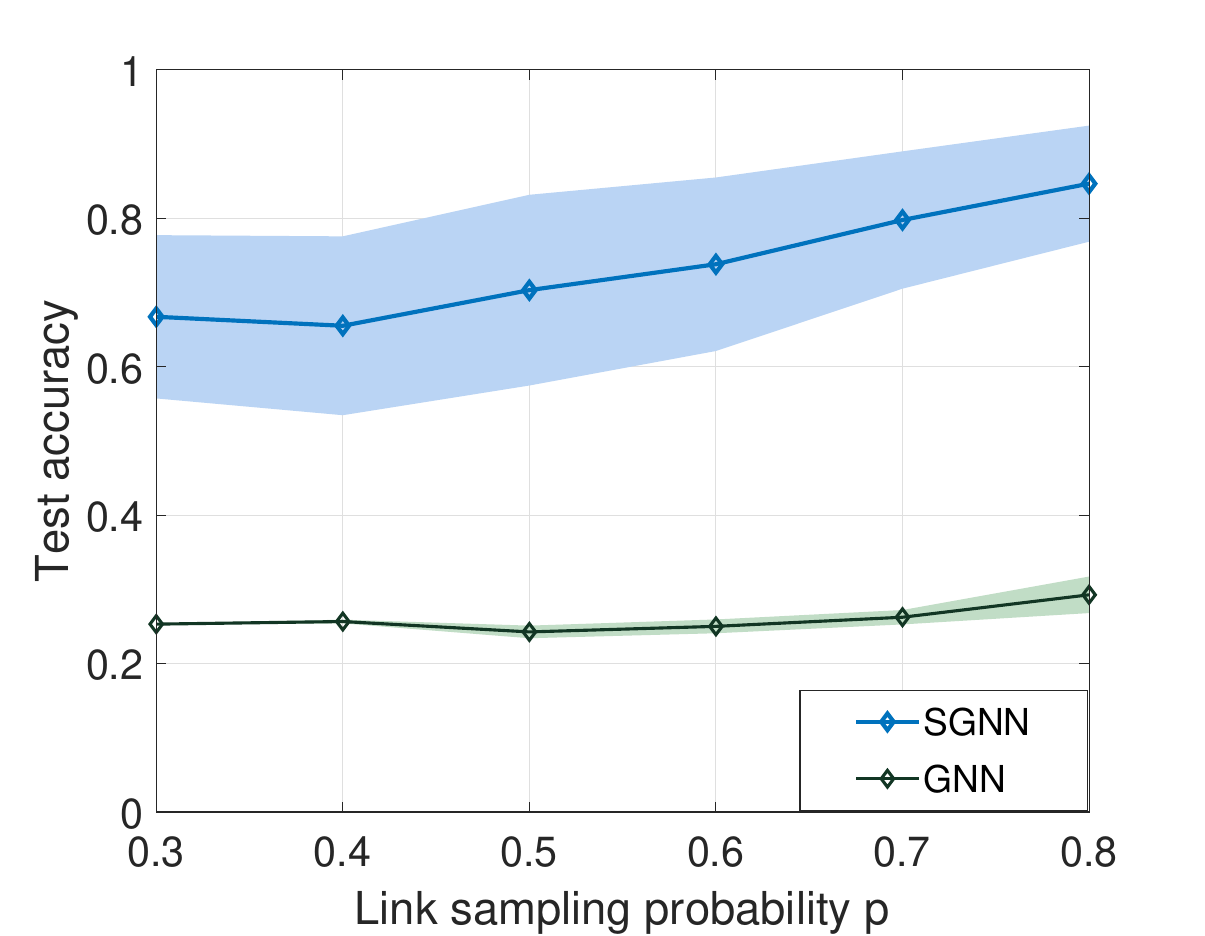}%
\caption{}%
\label{subfig:source_vary_p}%
\end{subfigure}\hfill\hfill%
\begin{subfigure}{0.6\columnwidth}
\includegraphics[width=1.075\linewidth,height = 0.85\linewidth]{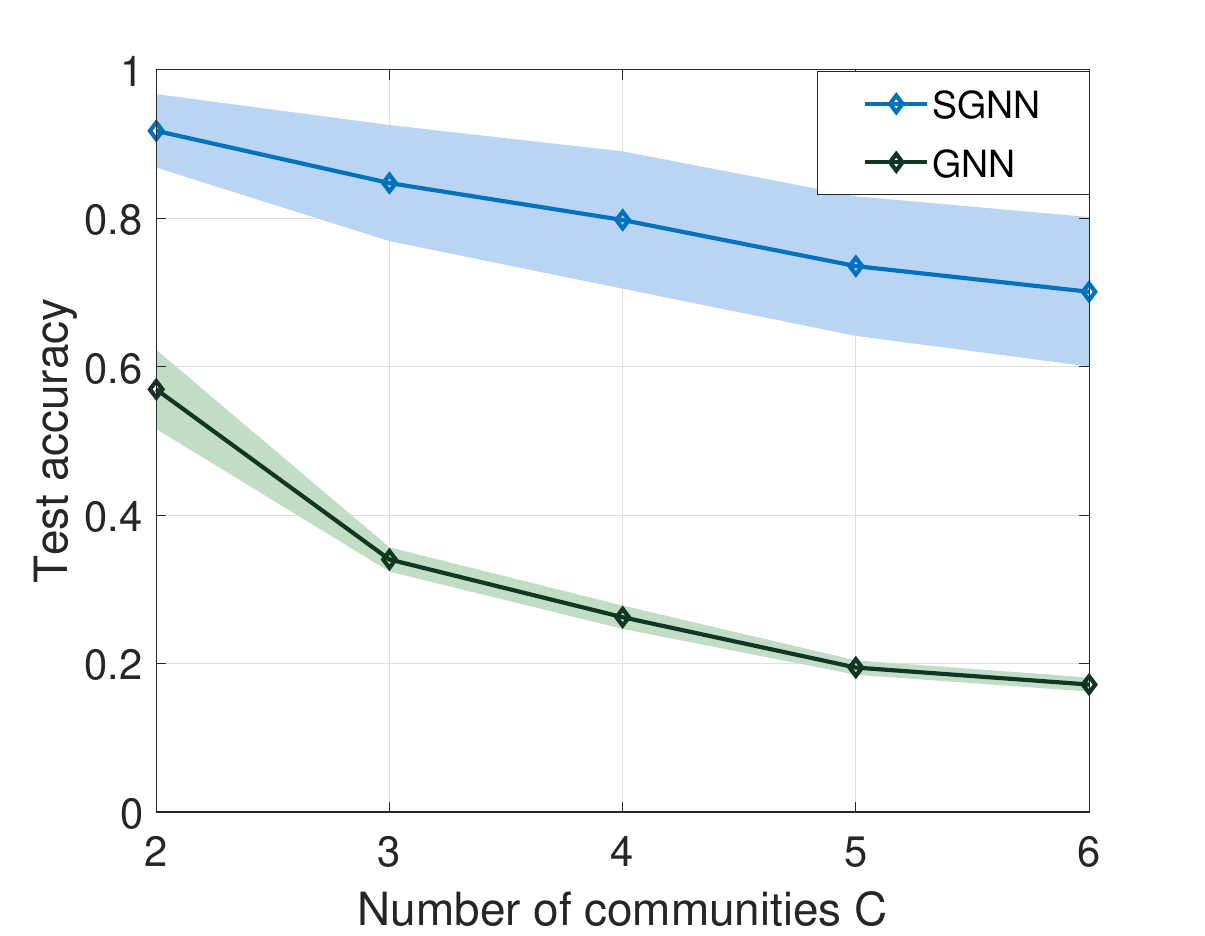}%
\caption{}%
\label{subfig:source_vary_C}%
\end{subfigure}%
\caption{Stochastic graph neural network performance for source localization. (a) Comparison between SGNN and GNN under different link sampling probabilities $p \in [0.85,0.95]$. (b) Comparison between SGNN and GNN under different link sampling probabilities $p \in [0.2,0.8]$. (c) Comparison between SGNN and GNN under different number of communities.}\label{fig:source}
\end{figure*}

Overall, we defined the stochastic graph neural network in \eqref{eq.sgnn} as an GNN architecture composed of stochastic graph filters. The SGNN output is random due to the stochastic architecture and thus the training shall optimize the SGNN with respect to average performance. In Section \ref{sec:expper}, we quantified how far from the mean a SGNN realization can behave and the effects played by different factors such as the shift operator, the filter type and the architecture structure. We also developed the SGD based learning process for the SGNN and the result \eqref{eq:thm22} of Theorem \ref{theorem2} shows that this learning process is sufficient to reach a local minima in expectation.


\section{Numerical Simulations} \label{Numerical}

We evaluate the proposed model and compare it with the convolutional GNN \cite{Fernando2019} on source localization (Section \ref{source}) and robot swarm control (Section \ref{robot}), and corroborate the variance bound in Theorem \ref{theorem1} numerically (Section \ref{subsec:varianceBound}). To train the architectures, we used the ADAM optimizer with decaying factors $\beta_1 = 0.9$ and $\beta_2 = 0.999$ \cite{Ba2010}. In the test phase, we assume all links may fall according to the RES($\ccalG,p$) model.

\subsection{Source Localization} \label{source}

We consider a signal diffusion process over a stochastic block model (SBM) graph of $N=40$ nodes divided equally into $C=4$ communities, with inter- and intra-block edge probability of $0.8$ and $0.2$ respectively. The goal is for a single node to find out distributively which community is the source of a given diffused signal. The initial source signal is a Kronecker delta $\bbdelta_{c} \in \reals^N$ centered at the source node $\{ n_c \}_{c=1}^C$ and diffused at time $\tau$ as $\bbx_{\tau c} =\bbS^\tau \bbdelta_{c} + \bbn$ with $\bbS = \bbA / \lambda_{max}(\bbA)$ and $\bbn \in \mathbb{R}^N$ a zero-mean Gaussian noise. 

We considered a one-layer SGNN with $32$ parallel filters of order $K = 10$ and ReLU nonlinearity. The learning rate is $\alpha = 10^{-3}$ with the mini-batch size of $1,000$ samples. The training set comprises $10^4$ tuples $\{ (\bbx_{\tau c}, n_c) \}$ picked uniformly at random for $\tau \in \{0, \ldots, 40\}$ and $n_c \in \{1, \ldots, 40\}$; the validation set contains $2,400$ of these tuples; the test set contains $1,000$ tuples. Our results are averaged over ten different data and ten different graph realizations for a total of $100$ Monte-Carlo runs.

\smallskip
\noindent\textbf{Convergence analysis.} We first corroborate the convergence analysis in Section \ref{sec:convergence} and show the SGNN approaches a stationary point. Figure \ref{fig:convergence} shows the learning process of the SGNN with link sampling probabilities $p=0.9$, $p=0.7$, and $p=0.5$. These values correspond to stable, relatively stable, and vulnerable scenarios, respectively. The cost value decreases with the number of iterations, leading to a convergent result in all cases. When $p=0.9$, the SGNN exhibits the best behavior and converges to a lower value. This is because of the higher link stability. As $p$ decreases indicating more graph randomness, the convergent value increases accordingly as observed the lines corresponding to $p=0.7$ and $p=0.5$. The convergent values of $p=0.7$ and $p=0.5$ also have larger varying errors, which can be explained by the increasing stochastic error induced by the increased network randomness. These errors can be further reduced by either decreasing the step-size or training the network longer.

\begin{figure*}%
\centering
\begin{subfigure}{0.61\columnwidth}
\includegraphics[width=1.1\linewidth, height = 0.85\linewidth]{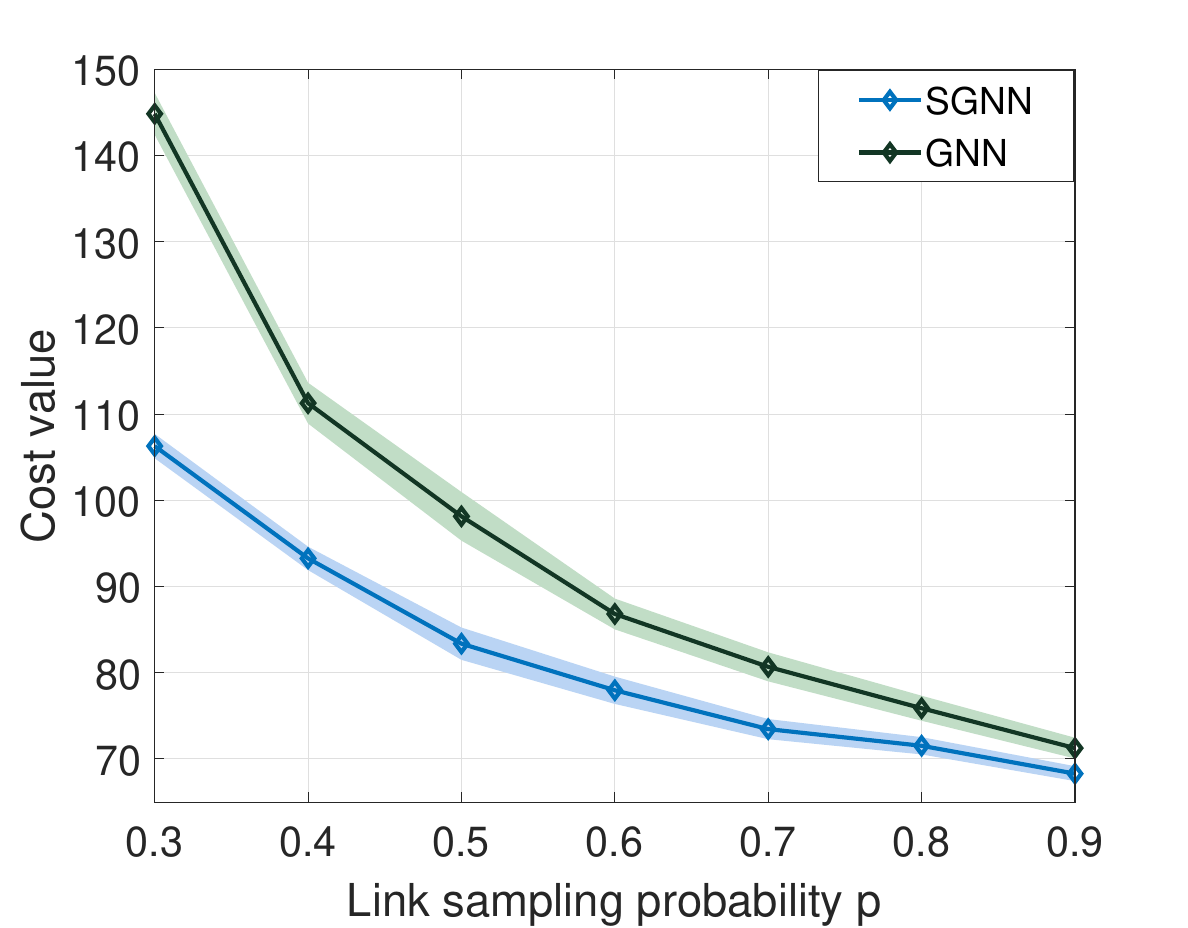}%
\caption{}%
\label{subfig:robot_vary_p}%
\end{subfigure}\hfill\hfill%
\begin{subfigure}{0.61\columnwidth}
\includegraphics[width=1.1\linewidth,height = 0.85\linewidth]{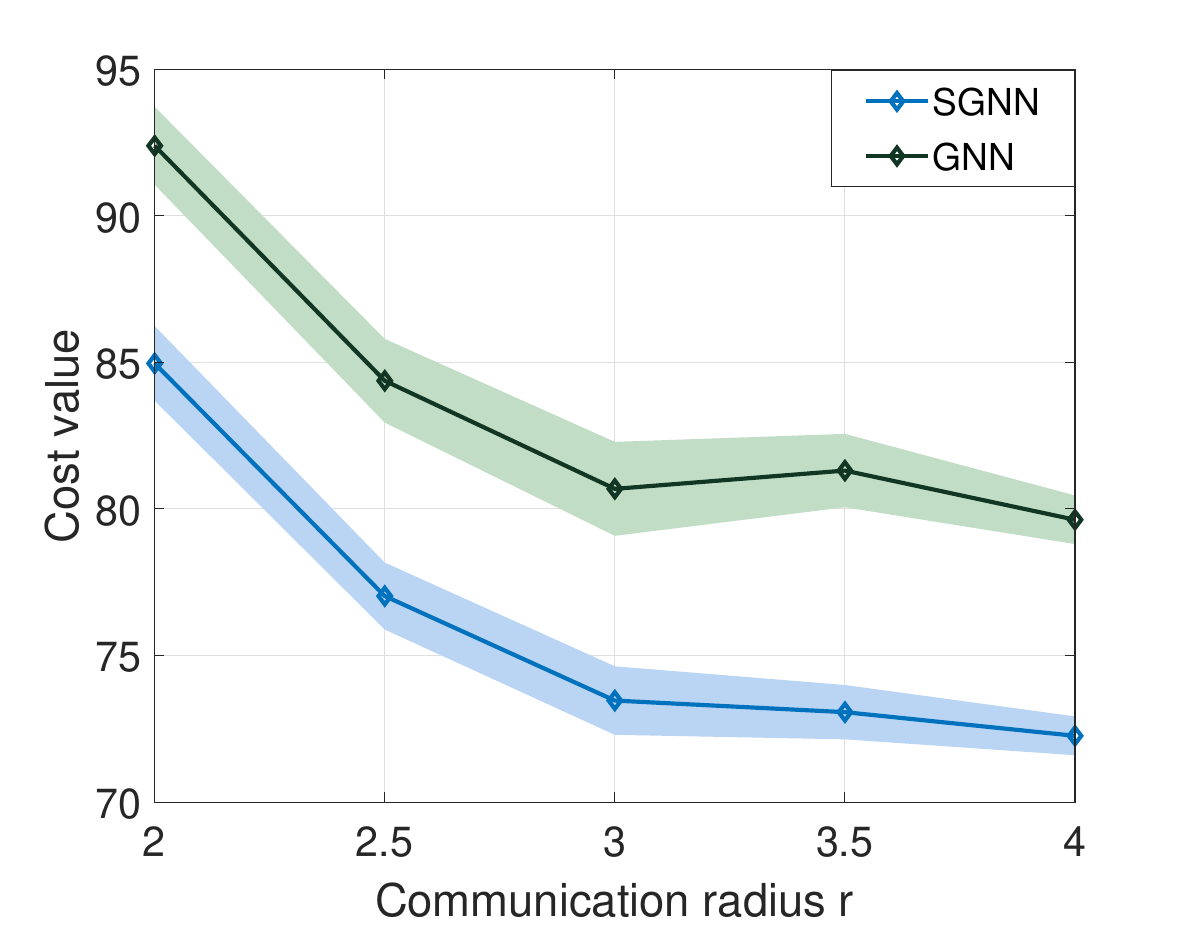}%
\caption{}%
\label{subfig:robot_vary_r}%
\end{subfigure}\hfill\hfill%
\begin{subfigure}{0.6\columnwidth}
\includegraphics[width=1.075\linewidth,height = 0.85\linewidth]{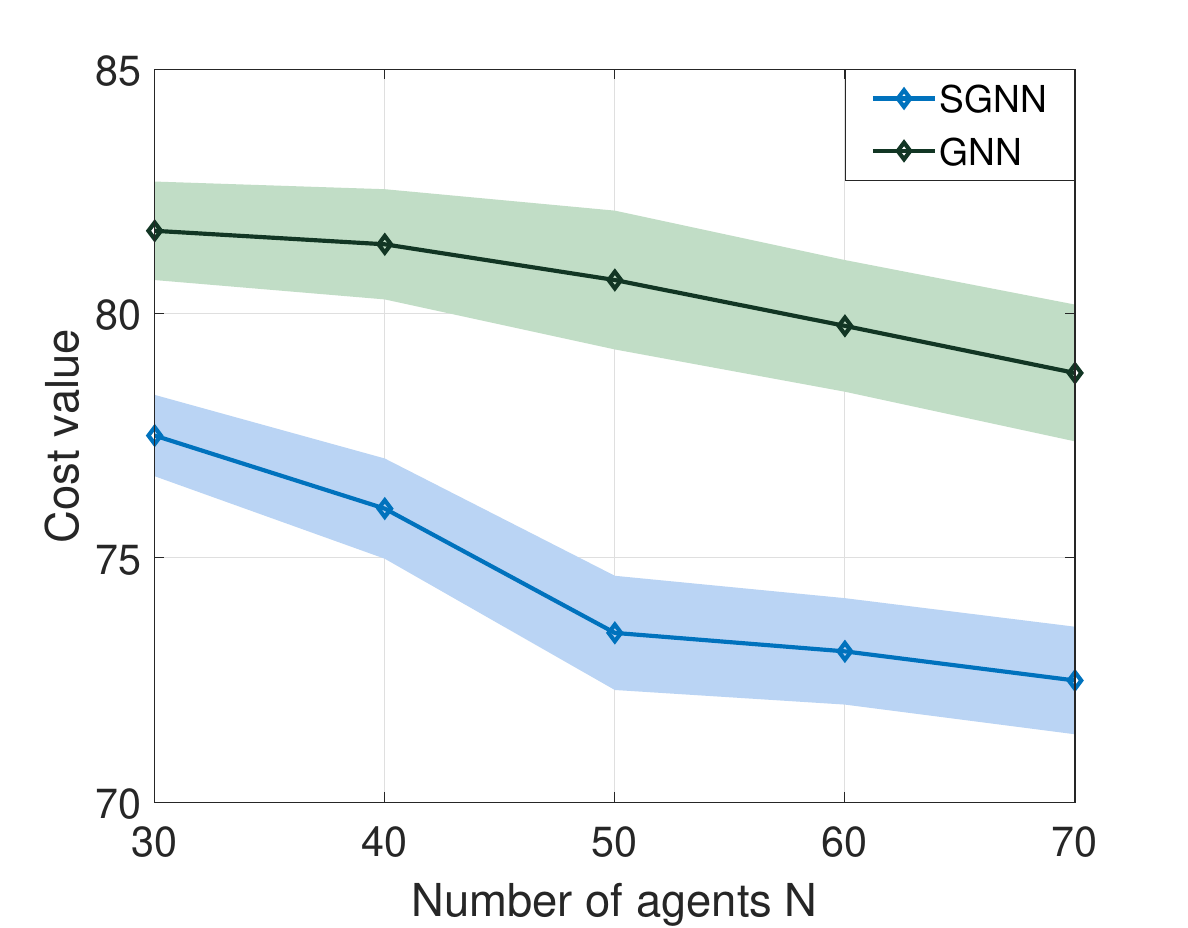}%
\caption{}%
\label{subfig:robot_vary_n}%
\end{subfigure}%
\caption{Stochastic graph neural network performance for robot swarm control. (a) Comparison between the SGNN and the GNN under different link sampling probabilities $p \in [0.2,0.9]$. (b) Comparison between the SGNN and the GNN under different different communication radius. (c) Comparison between the SGNN and the GNN under different number of agents.}\label{fig:robot}
\end{figure*}

\smallskip
\noindent\textbf{Link probability.} We then compare the impact of link sampling probability on the SGNN with the GNN. Figure \ref{subfig:source_vary_p_zoom} shows the classification test accuracy over relatively stable topologies, i.e., $p \in [0.85, 0.95]$. The SGNN exhibits a higher accuracy and lower variance, which are emphasized for lower values of $p$. That is, the more links are lost, the more the SGNN outperforms the GNN. This highlights the importance of accounting for the topological randomness during training. When $p$ approaches one, the GNN shows a comparable test accuracy, indicating the GNN is a valid choice only for highly stable topologies. Figure \ref{subfig:source_vary_p} compares the two architectures for $p \in [0.3, 0.8]$, where links are lost often and the network varies dramatically. The SGNN maintains a good performance even in this severe scenario. Instead, the GNN losses entirely its discriminative power and yields a performance similar to a random classifier. We attribute the latter to the fact that the ideal fixed graph in the GNN training deviates substantially from practical random graphs encountered in the test phase.

\smallskip
\noindent\textbf{Graph setting.} Finally, we compare the SGNN with the GNN under different graph settings. We train the SGNN on networks divided equally into $C\in \{2,3,4,5,6\}$ communities with each community containing $10$ nodes. The link sampling probability is $p=0.7$. Figure \ref{subfig:source_vary_C} illustrates the SGNN outperforms the GNN in all scenarios. The GNN degrades towards a random classifier because the problem becomes more challenging as the number of communities $C$ increases. However, the SGNN maintains a good performance, which highlights the importance of robust transference.

\subsection{Robot Swarm Control} \label{robot}

The goal of this experiment is to learn a distributed controller for robot swarms to fly together and avoid collision \cite{Tolstaya2019}. We consider a network of $N$ agents, where each agent $i$ is described by its position $\bbz_{i}\in \mathbb{R}^2$, velocity $\bbv_{i}\in \mathbb{R}^2 $, and acceleration $\bbu_{i}\in \mathbb{R}^2$. The problem has an optimal centralized solution on accelerations
\begin{equation}\label{eq:theorem3main}
\begin{split}
\bbu^*_{i}=-\sum_{j=1}^N \left( \bbv_{i}-\bbv_{j} \right)-\sum_{j=1}^N \rho\left( \bbz_{i}, \bbz_{j} \right)
\end{split}
\end{equation}
that assigns each agent's velocity to the mean velocity. Here, $\rho\left( \bbz_{i}, \bbz_{j} \right)$ is the collision avoidance potential. The centralized controller is, however, not practical because it requires velocity and position information of all agents at each agent $i$. We aim to learn a distributed controller with GNN by relying only on local neighborhood information.

We consider agent $i$ communicates with agent $j$ if their distance $\| \bbz_{i} - \bbz_{j} \|_2 \le r$ is within the communication radius $r$. The communication graph $\ccalG = (\ccalV, \ccalE)$ involves the node set $\ccalV = \{ 1,\ldots,N \}$ as agents, the edge set $\ccalE$ as available communication links, and $\bbS$ as the associated graph shift operator. The graph signal $\bbx$ is the relevant feature designed with the agent position $\bbz$ and velocity $\bbv$ \cite{Tolstaya2019}. We measure the controller performance with the variance of velocities for a trajectory, which quantifies how far the system is from consensus in velocities \cite{Xiao2007}.

As baseline, we assume $N=50$ agents and a communication radius $r=3.0{\rm m}$. The agents are distributed randomly in a circle with a minimum separation of $0.1$m and initial velocities sampled uniformly in the interval $[-3.0{\rm m/s}, +3.0{\rm m/s}]$. We consider a one-layered SGNN with $32$ parallel filters of order $K=3$ and use tangent nonlinearity like in \cite{Tolstaya2019}. We use imitation learning to train the SGNN over a training set of $1000$ trajectories, each containing $100$ time steps. The validation and test sets contain each 100 extra trajectories. We train the SGNN for $30$ epochs with batch size of $20$ samples and learning rate $\alpha=3\cdot 10^{-4}$. Our results are averaged over $10$ simulations.

\smallskip
\noindent\textbf{Link probability.} We compare the performance of the SGNN with the GNN under different link sampling probabilities. Figure \ref{subfig:robot_vary_p} shows the cost value of the two architectures for $p \in [0.3,0.9]$. The SGNN achieves both a lower mean and variance compared with the GNN. This improved SGNN performance is more visible for lower link sampling probabilities $p$, i.e., as links become more unstable. We again attribute this behavior to the robust transference of the SGNN, since it accounts for link instabilities during training. However, notice the SGNN also degrades when the graph varies dramatically (small $p$). This is because the information loss induced by link fluctuations leads to inevitable errors, which cannot be resolved by training.

\begin{figure*}%
\centering
\begin{subfigure}{0.5\columnwidth}
\includegraphics[width=1.075\linewidth,height = 0.85\linewidth]{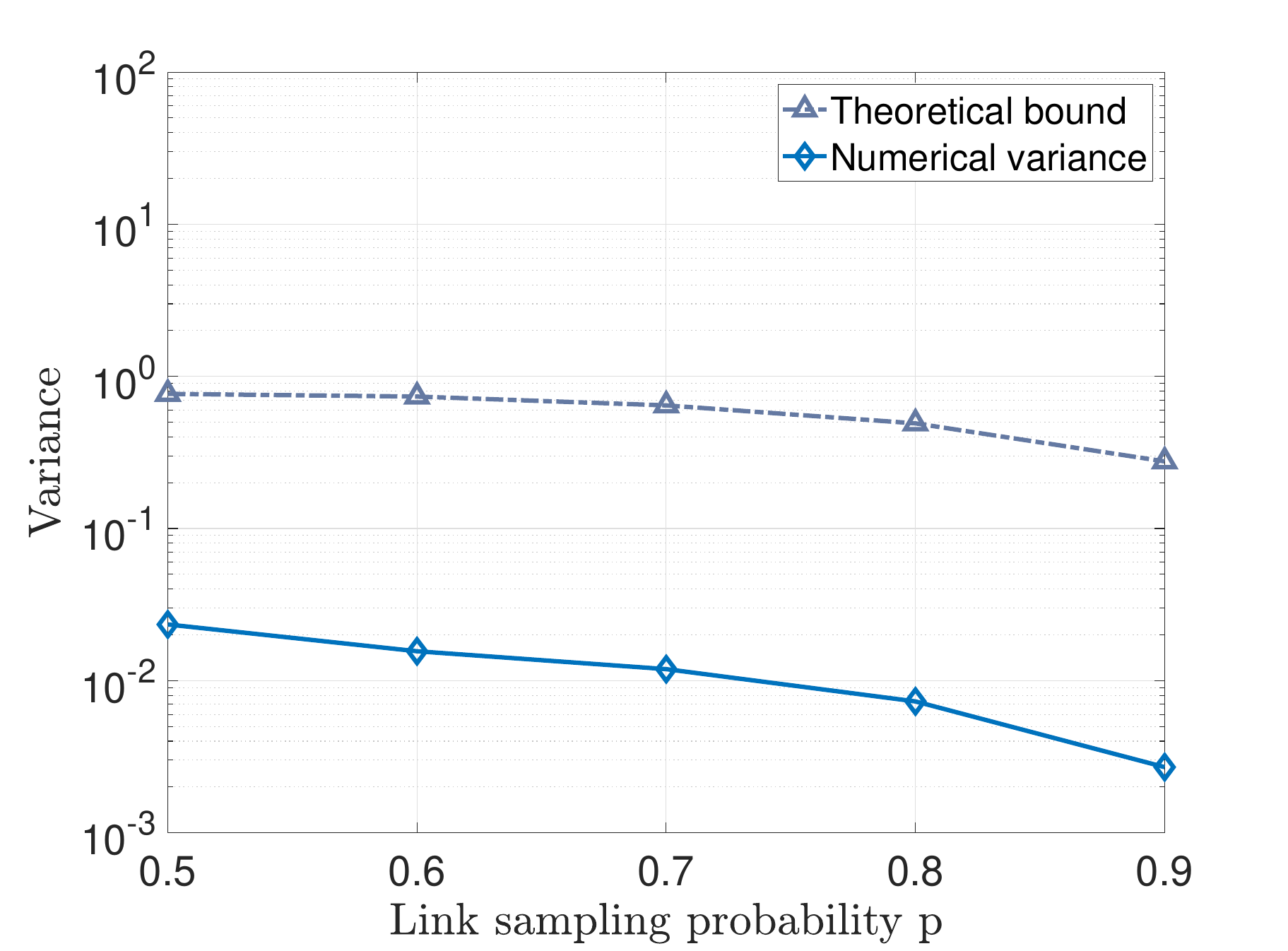}%
\caption{}%
\label{subfigd_differentp}%
\end{subfigure}%
\begin{subfigure}{0.5\columnwidth}
\includegraphics[width=1.1\linewidth, height = 0.85\linewidth]{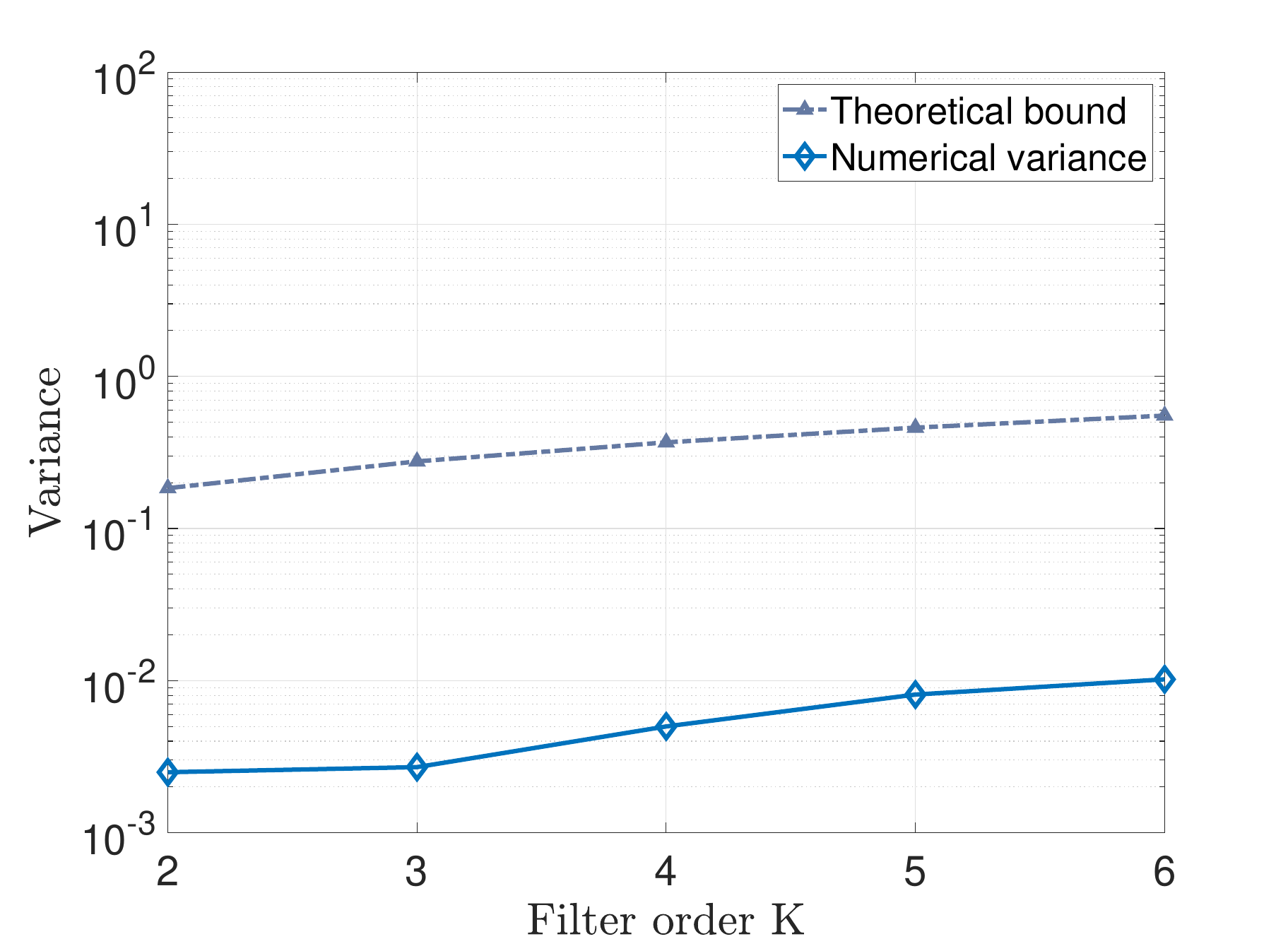}%
\caption{}%
\label{subfiga_differentK}%
\end{subfigure}\hfill\hfill%
\begin{subfigure}{0.5\columnwidth}
\includegraphics[width=1.1\linewidth,height = 0.85\linewidth]{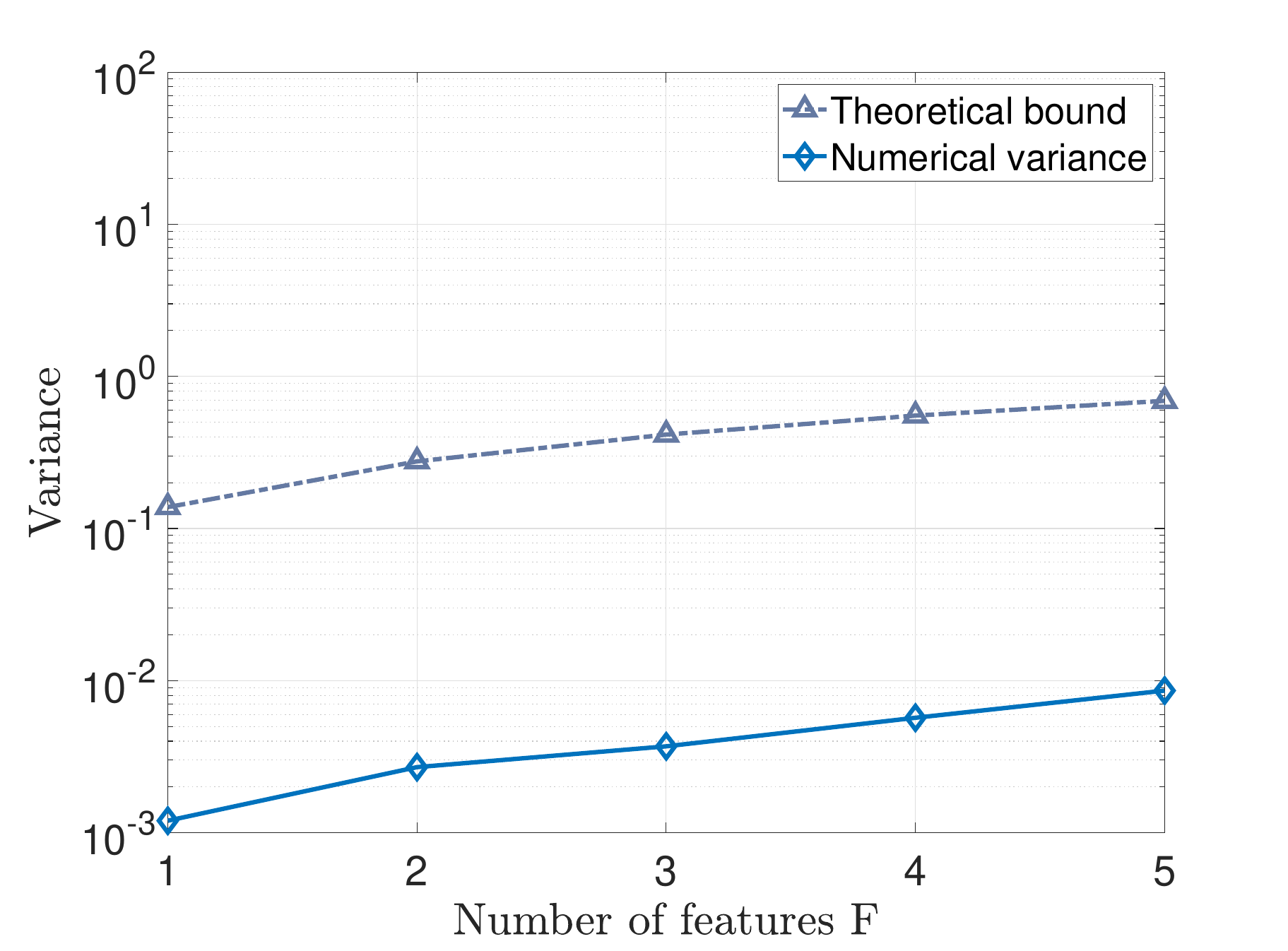}%
\caption{}%
\label{subfigb_differentF}%
\end{subfigure}\hfill\hfill%
\begin{subfigure}{0.5\columnwidth}
\includegraphics[width=1.075\linewidth,height = 0.85\linewidth]{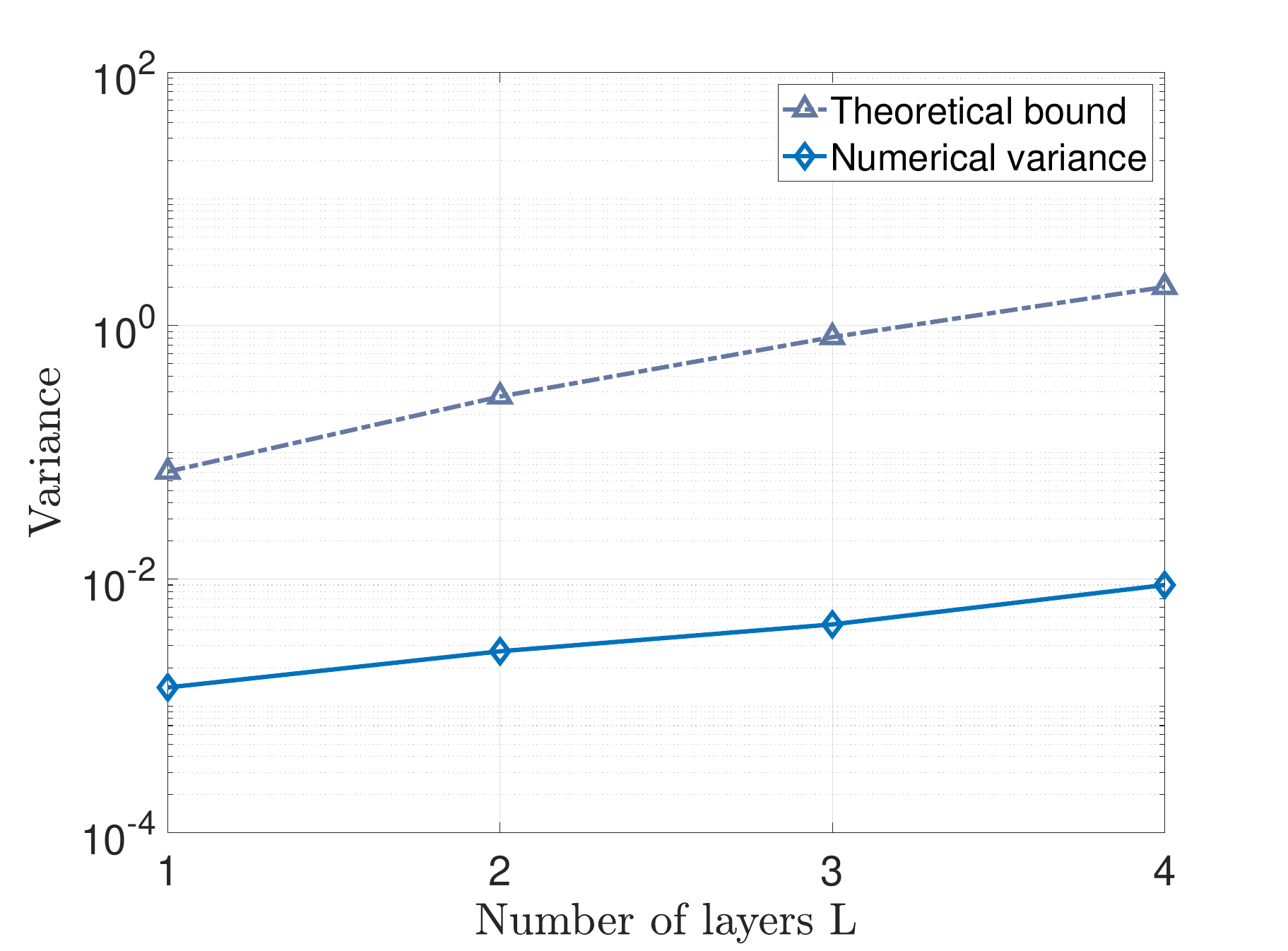}%
\caption{}%
\label{subfigc_differentL}%
\end{subfigure}%
\caption{Variance comparison between theoretical bound and empirical variance. (a) Different link sampling probabilities $p$. (b) Different filter orders $K$. (c) Different number of filters $F$ (d) Different number of layers $L$.}\label{fig:varianceBound}
\end{figure*}

\smallskip
\noindent\textbf{Communication radius.} Figures \ref{subfig:robot_vary_r} depicts the performance of the SGNN and the GNN for different communication radius $r$. The link sampling probability is $p=0.7$. In general, the SGNN outperforms the GNN in all scenarios. The SGNN performance improves as the communication radius $r$ increases. This follows our intuition because a larger communication radius increases the node exchanges, which helps contrasting some of the link losses. On the other hand, a denser graph with larger communication radius magnifies the impact of link losses by introducing more randomness, such that the decreasing rate of cost value reduces seen from $r=3$m to $r=4$m. The performance improvement gained from the SGNN increases in this case, which emphasizes the importance of accounting for the graph randomness during training. 

\smallskip
\noindent\textbf{Agent number.} Lastly, we compare two architectures for different number of agents $N$. In Figure \ref{subfig:robot_vary_n}, we see the cost decreases as the number of agents increases, which can be explained by the increased information exchanges in large networks. This result indicates the SGNN is capable of handling large-scale robot swarms while retaining a good performance. In addition, the SGNN improvement becomes more visible for larger networks since topological randomness has more effects in these cases and thus robust transference plays a more important role.

\subsection{Variance Corroboration}\label{subsec:varianceBound}

The goal of this experiment is to corroborate the variance bound in Theorem \ref{theorem1}. We consider the SBM graph of $N=40$ nodes and the SGNN of $L=2$ layers, each with $F=2$ filters per layer of order $K=3$ followed by the ReLU nonlinearity. The input graph signal $\bbx$ is with unitary energy and the link sampling probability is $p=0.9$. We compare the empirical variance and the theoretical bound under changing scenarios; namely, different link sampling probabilities $p$, different filter orders $K$, different numbers of features $F$, and different numbers of layers $L$.

Fig. \ref{fig:varianceBound} shows the results. First, we note that the theoretical analysis yields a fair bound for the numerical variance, and the results corroborate the impact of the graph stochasticity and architecture hyper-parameters on the variance of the SGNN output as indicated in Theorem \ref{theorem1}. That is, the variance decreases as the link sampling probability increases from $0.5$ to $0.9$, and it increases with the filter order $K$, the number of features $F$, and the number of layers $L$. We also note that the theoretical bound is not tight, essentially because this bound holds uniformly for all graphs and thus is not tight. The main goal of variance analysis is to show that the SGNN is statistically stable and indicates the role played by the graph stochasticity and different architecture hyper-parameters on the SGNN output variance.


\section{Conclusions} \label{sec_conclusions}

We proposed a distributed stochastic graph neural network that can operate over random time varying topologies. The architecture is similar to the conventional GNNs but substitutes the convolutional graph filters with stochastic graph filters to account for the topological randomness during training. A detailed mathematical analysis characterized the output variance of the SGNN, which is upper bounded by a factor that is quadratic in the link sampling probability, indicating the impact of link losses to the SGNN output. We further formulated a learning process that accounts for the randomness in the cost function and leveraged stochastic gradient descent to prove this learning process converges to a stationary point. Numerical results corroborated the proposed theoretical model on distributed source localization and decentralized robot swarm control, showing superior performance compared with the GNN that ignores link losses. In near future, we plan to corroborate these findings with other topology variation models and in other distributed applications, such as power outage prediction in smart grids.


\appendices 


\section{Proof of Proposition 1} \label{pr:proposition1}

We need the following lemma, whose proof is in the supplementary material.

\begin{lemma} \label{lemmma3}
	Consider the underlying graph $\ccalG$ with the shift operator $\bbS$ and let $\bbS_k$ be the shift operator of $k$th RES($\ccalG,p$) realization of $\bbS$, $\barbS=\mathbb{E}[\bbS_k]$ the expected shift operator, and $\bbD$ the diagonal degree matrix with $d_i$ the degree of node $i$. Then, it holds that
	\begin{equation}\label{prlem33}
		\mathbb{E}\left[ \bbS_k^2 \right] = 
		\begin{cases}
			\barbS^2 + p (1-p)\bbD,\! & \!\text{if } \bbS = \bbA, \\
			\barbS^2 + 2 p (1-p)\bbS,\! & \!\text{if } \bbS = \bbL
		\end{cases}
	\end{equation}
	with $\bbA$ the adjacency matrix and $\bbL$ the Laplacian matrix.  
\end{lemma}

\begin{proof}[Proof of Proposition \ref{proposition1}] 
	
	Let $\bbu = \bbH(\bbS_{K:0})\bbx$ and $\barbu = \bbH(\barbS)\bbx$ be the output and the expected output of stochastic graph filter. Substituting the filter expression into the variance, we get\! \footnote{Throughout this proof, we will use the shorthand notation $\sum_{a,b,c = \alpha, \beta, \gamma}^{A, B, C} (\cdot)$ to denote $\sum_{a = \alpha}^A\sum_{b = \beta}^B\sum_{c = \gamma}^C (\cdot)$ to avoid overcrowded expressions. When the extremes of the sum ($\alpha, \beta, \gamma$ or $A, B, C$) are the same, we will write directly the respective value.}
	\begin{gather} \label{eq:prop331}
		\begin{split}
			&{\rm var}[\bbu] = \mathbb{E} \left[ \tr \left( \bbu \bbu^{\rm H} - \bar{\bbu} \bar{\bbu}^{\rm H} \right) \right]\\
			& = \sum_{k, \ell=0}^K h_k h_\ell \left( \mathbb{E} \left[ \tr \left( T(k,\ell) \right) \right] -  \mathbb{E} \left[ \tr \left( \bar{T}(k,\ell) \right) \right] \right)
		\end{split}
	\end{gather}
	where $T(k,\ell)= \bbS_{k:0} \bbx \bbx^{\rm H} \bbS_{0:\ell}$ and $\bar{T}(k, \ell)= \bar{\bbS}^k  \bbx \bbx^{\top} \bar{\bbS}^\ell$. To further simplify notation, we denote with $\lceil k\ell \rceil = \max(k,\ell)$ and $\lfloor k\ell \rfloor = \min(k,\ell)$. Expression \eqref{eq:prop331} holds because of the linearity of the trace and expectation and the symmetry of the shift operators $\bbS_k$ and $\barbS$. Since $\bbS_k$ is a realization, we represent it as $\bbS_k = \barbS + \bbE_k$ where $\bbE_k$ is the deviation of $\bbS_k$ from the mean $\barbS$. Substituting this result into $T(k,\ell)$ and since $\mathbb{E}[\bbE_k] = \bb0$, we have
	\begin{align}
		\label{eq:prop33Ts1} &\mathbb{E}\left[ T(k,\ell)\right] = \mathbb{E}\left[ (\barbS + \bbE_k)\cdots \bbx \bbx^\top \cdots (\barbS + \bbE_\ell)\right]\\
		&={\barbS}^{k}\bbx \bbx^\top {\barbS}^{\ell} \!+\! \mathbb{E}\!\left[\! \sum_{r=1}^{\lfloor k\ell \rfloor} \barbS^{k\!-\!r}\bbE_r \barbS^{r\!-\!1}\bbx \bbx^\top \barbS^{r\!-\!1}\bbE_r \barbS^{\ell\!-\!r}\!\right] \!+\! \mathbb{E}\!\left[\! \bbC_{k\ell}\!\right]. \nonumber 
	\end{align}
	The first term in \eqref{eq:prop33Ts1} yields from the maximum powers of $\barbS$ of the products; the second term captures all cross-products where we should note that for $k \ne \ell$ we have $\mathbb{E}\left[\bbE_k \bbE_\ell\right] = \mathbb{E}[\bbE_k]\mathbb{E}[\bbE_\ell] = \bb0$ due to independence and also the terms $r > \lfloor kl\rfloor$ are null due to the presence of a single expectation $\mathbb{E}[\bbE_r]$; the third term $\bbC_{k\ell}$ collects the sum of the remaining terms. By substituting $\tr \left( \mathbb{E}\left[ \bar{T}(k,\ell)\right] \right) = \tr \left( {\barbS}^{k}\bbx \bbx^\top {\barbS}^{\ell} \right)$ and \eqref{eq:prop33Ts1} into \eqref{eq:prop331}, we have
	\begin{align} \label{eq:prop332}
		&{\rm var}[\bbu] =\sum_{k=0}^K \sum_{\ell=0}^K h_k h_\ell \tr \left( \mathbb{E}\left[ \bbC_{k\ell}\right] \right)\\
		& \! +\!\sum_{k=1}^K \sum_{\ell=1}^K h_k h_\ell \tr \left( \mathbb{E}\!\left[\! \sum_{r=1}^{\lfloor k\ell \rfloor} \barbS^{k-r}\bbE_r \barbS^{r\!-\!1}\bbx \bbx^\top \barbS^{r-1}\bbE_r \barbS^{\ell-r}\!\right] \right)\nonumber.
	\end{align}
	We now analyze the two terms in \eqref{eq:prop332} separately. For this analysis, we will need the inequality
	\begin{gather} \label{eq:prop3ineq0}
		\begin{split}
			{\rm tr} (\bbA \bbB) \le \frac{\| \bbA + \bbA^\top \|_2}{2}{\rm tr}(\bbB) \le \| \bbA \|_2 {\rm tr}(\bbB)
		\end{split}
	\end{gather}
	that holds for any square matrix $\bbA$ and positive semi-definite matrix $\bbB$ \cite{Wang1986}.
	
	$\textbf{Second term.}$ By bringing the trace inside the expectation (due to their linearity) and leveraging the trace cyclic property $\tr(\bbA\bbB\bbC) = \tr(\bbC\bbA\bbB) = \tr(\bbB\bbC\bbA)$, we can write
	\begin{align}
		\label{eq:prop33} &\mathbb{E}\left[ \tr \left( \sum_{k,\ell=1}^K h_k h_\ell \sum_{r=1}^{\lfloor k\ell \rfloor} \barbS^{k-r}\bbE_r \barbS^{r-1}\bbx \bbx^\top \barbS^{r-1}\bbE_r \barbS^{\ell-r} \right) \right]\nonumber \\
		& \!=\!\mathbb{E}\!\left[ \! \sum_{r=1}^K \tr\! \left( \!\sum_{k, \ell=r}^K \! h_k h_\ell \bbE_r \barbS^{k+\ell-2r}\bbE_r \barbS^{r-1}\bbx \bbx^\top\! \barbS^{r-1}\! \right)\!\right]\!.
	\end{align}
	Notice that in \eqref{eq:prop33} we also rearranged the terms to change the sum limits. Since both matrices $\sum_{k, \ell=r}^K h_k h_\ell \bbE_r \barbS^{k+\ell-2r}\bbE_r = \big(\sum_{k=r}^K h_k \barbS^{k-r}\bbE_r\big)^\top \big(\sum_{k=r}^K h_k \barbS^{k-r}\bbE_r\big)$ and $\barbS^{r-1}\bbx \bbx^\top \barbS^{r-1}$ are positive semi-definite, we can use the Cauchy-Schwarz inequality $\tr(\bbA \bbB) \le \tr(\bbA)\tr(\bbB)$ \cite{Zhang1999} to upper bound \eqref{eq:prop33} by
	\begin{align}
		\label{eq:prop34} \!\mathbb{E}\!\left[ \!\sum_{r\!=\!1}^K\! \sum_{k, \ell=r}^K \!h_k h_\ell \tr\! \left( \bbE_r \barbS^{k+\ell-2r}\bbE_r\! \right)\! \tr\! \left( \barbS^{r\!-\!1}\bbx \bbx^\top \barbS^{r\!-\!1} \!\right)\!\right]\!.
	\end{align}
	
	We now proceed by expressing the graph signal $\bbx$ in the frequency domain of the expected graph. Let $\barbS = \barbV\barbLambda\barbV^\top$ be the eigendecomposition of $\barbS$ with eigenvectors $\barbV = [\barbv_1,\ldots,\barbv_N]^\top$ and eigenvalues $\barbLambda = \text{diag}(\bar{\lambda}_1, \ldots, \bar{\lambda}_N)$. Substituting the graph Fourier expansion $\bbx = \sum_{i=1}^N \hat{x}_i \barbv_i$ into $\tr\! \left( \barbS^{r-1}\bbx \bbx^\top \barbS^{r-1} \right)$, we get
	\begin{gather}\label{eq:prop35}
		\tr\! \left( \barbS^{r\!-\!1}\bbx \bbx^\top \barbS^{r\!-\!1} \right) \!=\! \sum_{i\!=\!1}^N\! \hat{x}_i^2 \bar{\lambda}_i^{2r\!-\!2} \tr\! \left(\! \barbv_i \barbv_i^\top\! \right) \!=\! \sum_{i\!=\!1}^N \hat{x}_i^2 \bar{\lambda}_i^{2r\!-\!2}
	\end{gather}
	where $\tr(\bbv_i\bbv_i^\top) = 1$ for $i=1,\ldots,N$ due to the orthonormality of eigenvectors. By substituting \eqref{eq:prop35} into \eqref{eq:prop34}, we get
	\begin{align}
		\label{eq:prop36} \!\sum_{i\!=\!1}^N \hat{x}_i^2 \mathbb{E}\!\left[ \! \sum_{r\!=\!1}^K \sum_{k, \ell=r}^K\! h_k h_\ell \bar{\lambda}_i^{2r-2} \tr\! \left( \bbE_r \barbS^{k+\ell-2r}\bbE_r \right)\!\right].
	\end{align}
	Using again the trace cyclic property to write $\tr(\bbE_r\barbS^{k+\ell-2}\bbE_r) = \tr(\barbS^{k+\ell-2}\bbE_r^2)$ and the linearity of the expectation, we have
	\begin{align} \label{eq:prop37} 
		& \sum_{i = 1}^N\hat{x}_i^2\! \sum_{r\!=\!1}^K\!\tr\! \left( \sum_{k, \ell=r}^K h_k h_\ell \bar{\lambda}_i^{2r-2} \barbS^{k+\ell-2r} \mathbb{E}\!\left[ \bbE_r^2\right] \right).
	\end{align}
	From Lemma \ref{lemmma3}, we have $\mathbb{E}\left[ \bbE_k^2 \right] =  \alpha p(1-p) \bbE$ with $\alpha = 1$ and $\bbE = \bbD$ if $\bbS = \bbA$ and $\alpha = 2$ and $\bbE = \bbS$ if $\bbS = \bbL$. By substituting this result into \eqref{eq:prop37} and using inequality \eqref{eq:prop3ineq0} since $\bbE$ is positive semi-definite, we have
	\begin{align}
		\label{eq:prop39}
		& \sum_{i = 1}^N\hat{x}_i^2\! \sum_{r\!=\!1}^K\!\tr\! \left( \sum_{k, \ell=r}^K h_k h_\ell \bar{\lambda}_i^{2r-2} \barbS^{k+\ell-2r} \left( \alpha p(1-p) \bbE\right)\! \right)\\
		& \le \!\alpha p(1-p)\sum_{i = 1}^N\hat{x}_i^2 \big\| \sum_{r\!=\!1}^K\! \sum_{k, \ell=r}^K\! h_k h_\ell \bar{\lambda}_i^{2r-2} \barbS^{k+\ell-2r} \big\|_2 \tr\! \left( \bbE \right) \nonumber
	\end{align}
	with $\tr\left( \bbE \right)=\sum_{i=1}^N d_i =2M$ and $M$ the number of edges. 
	
	At this point, we proceed to upper bound the filter matrix norm in \eqref{eq:prop39}. A standard procedure to bound the spectral norm of a matrix $\bbA$, is to upper bound the norm of $\|\bbA\bba\|_2$ as $\|\bbA\bba\|_2 \le A \|\bba\|_2$ for any vector $\bba$ \cite{Meyer2000}. In this instance, $A$ is the upper bound for the norm of $\bbA$. Following this rationale, we consider the GFT expansion of a vector $\bba$ on the expected graph $\bba = \sum_{j=1}^N \hat{a}_j \barbv_j$ where $\{ \barbv_j \}_{j=1}^N$ are orthonormal. Then, we have
	\begin{align}
		\label{eq:prop310} & \!\big\| \!\sum_{r\!=\!1}^K\!\! \sum_{k, \ell\!=\!r}^K\! h_k\! h_\ell\! \bar{\lambda}_i^{2r\!-\!2}\! \barbS^{k\!+\!\ell\!-\!2r}\! \bba \big\|_2^2 \!\!=\!\! \sum_{j\!=\!1}^N\!\! \hat{a}_j^2 \big|\! \sum_{r\!=\!1}^K\!\! \sum_{k, \ell\!=\!r}^K\! h_k\! h_\ell\! \bar{\lambda}_i^{2r\!-\!2}\! \bar{\lambda}_j^{k\!+\!\ell\!-\!2r}\big|^2\!\!.
	\end{align}
	Consider now the expression inside the absolute value in \eqref{eq:prop310}. This expression is linked to the partial derivative of the generalized frequency response $h(\bblambda)$ in \eqref{eq:h2}. To detail this, we introduce the first-order partial derivative of the generalized frequency response $h(\bblambda)$ w.r.t. the $r$th entry $\lambda_r$ of $\bblambda$
	\begin{align}
		\label{eq:prop31051} 
		\frac{\partial h(\bblambda)}{\partial \lambda_r} \!=\! \sum_{k=r}^K h_k \lambda_{K:(r+1)} \lambda_{(r-1):1},\! ~\forall~r\!=\!1,\!\ldots,\!K
	\end{align}
	where $\lambda_{K:(r+1)} = \lambda_K \cdots \lambda_{r+1}$ and $\lambda_{(r-1):1}=\lambda_{r-1}\cdots \lambda_1$. Let us then consider $K$ specific eigenvalue vectors of dimensions $K\times 1$: $\barblambda_{ij}^1 = [\bar{\lambda}_j, ..., \bar{\lambda}_j]^\top$, $\barblambda_{ij}^2 = [\bar{\lambda}_i, \bar{\lambda}_j, \ldots, \bar{\lambda}_j]^\top$, $\ldots$ , $\barblambda_{ij}^K = [\bar{\lambda}_i, ...,\bar{\lambda}_i, \bar{\lambda}_j]^\top$ for two eigenvalues $\bar{\lambda}_i$ and $\bar{\lambda}_j$ of $\barbS$, and their respective generalized frequency responses $\{ h(\barblambda_{ij}^1), \ldots, h(\barblambda_{ij}^K) \}$\!\!\! \footnote{The generalized frequency response $h(\bblambda)$ in \eqref{eq:frere1} is an analytic function of the vector variable $\bblambda = [\lambda_1, \ldots, \lambda_K]^\top$ such that $\bblambda$ can take any value.} [cf. \eqref{eq:frere1}]. The $r$th first-order partial derivative $\partial h(\bblambda)/\partial \lambda_r$ of the generalized frequency response instantiated on $\barblambda_{ij}^r$ is
	\begin{align}
		\label{eq:prop3105} 
		\frac{\partial h(\barblambda_{ij}^r)}{\partial \lambda_r} = \sum_{k=r}^K h_k \bar{\lambda}_{i}^{r-1} \bar{\lambda}_j^{k-r},\! ~\forall~r\!=\!1,\!\ldots,\!K.
	\end{align}
	We then observe the expression inside the absolute value in \eqref{eq:prop310} can be represented as the sum of $K$ first-order partial derivatives; i.e., we can write it in the compact form
	\begin{align}
		\label{eq:prop311} &\sum_{r\!=\!1}^K\! \sum_{k, \ell=r}^K\! h_k h_\ell \bar{\lambda}_i^{2r-2} \bar{\lambda}_j^{k+\ell-2r} =  \sum_{r\!=\!1}^K \left( \frac{\partial h(\barblambda_{jr})}{\partial\lambda_{r}} \right)^2.
	\end{align}
	From Assumption \ref{as:4}, the generalized frequency responses are Lipschitz with constant $C_g$. Thus we can upper bound \eqref{eq:prop311} as
	\begin{align}
		\label{eq:prop3115} & \big| \sum_{r\!=\!1}^K\! \sum_{k,\ell=r}^K\! h_k h_\ell \bar{\lambda}_i^{2r-2} \bar{\lambda}_j^{k+\ell-2r}\big|^2  \le K^2 C_g^4
	\end{align}
	which implies the norm of the filter matrix in \eqref{eq:prop310} is upper bounded by $KC_g^2$.
	
	By substituting this norm bound into \eqref{eq:prop310} and altogether into \eqref{eq:prop33}, we have
	\begin{align}
		\label{eq:prop312} &\mathbb{E}\left[ \sum_{k, \ell=1}^K h_k h_\ell \sum_{r=1}^{\lfloor k\ell \rfloor} \tr \left( \barbS^{k-r}\bbE_r \barbS^{r-1}\bbx \bbx^\top \barbS^{r-1}\bbE_r \barbS^{\ell-r} \right) \right]\nonumber \\
		& \!\le\! 2\alpha MKC_g^2 \!\sum_{i\!=\!1}^N \!\hat{x}_i^2 p(1\!-\!p) \!=\! 2\alpha MKC_g^2\| \bbx \|^2_2 p(1\!-\!p).
	\end{align}
	
	$\textbf{First term.}$ Matrix $\bbC_{k\ell}$ comprises the sum of the remaining expansion terms. Each of these terms is a quadratic form in the error matrices $\bbE_k$, $\bbE_\ell$ with $k \neq l$; i.e., it is of the form $f_1(\barbS, h_k)\bbE_k f_2(\barbS, h_k)\bbE_kf_3(\barbS, h_k)\bbE_\ell f_4(\barbS, h_k)\bbE_\ell$ for some functions $f_1(\cdot), ..., f_4(\cdot)$ that depend on the expected shift operator and filter coefficients. Each of these double-quadratic terms can be bounded by a factor containing at least two terms $\tr\left(\mathbb{E}[\bbE_{k_1}^2]\right)$ and $\tr\left(\mathbb{E}[\bbE_{k_2}^2]\right)$. Since the frequency response $h(\lambda)$ is bounded from Assumption \ref{as:3}, also the coefficients $\{ h_k \}_{k=0}^K$ are bounded. Further since $\| \barbS \|_2$ is bounded and $\mathbb{E}[\bbE_{k}^2] = \alpha p(1-p)\bbE$ from Lemma \ref{lemmma3}, we can write the first term in \eqref{eq:prop33Ts1} as
	\begin{gather} \label{eq:prop3125}
		\begin{split}
			\mathbb{E} \left[\sum_{k, \ell=0}^K h_k h_\ell \bbC_{k\ell} \right] = \ccalO(p^2(1-p)^2)\|\bbx\|_2^2.
		\end{split}
	\end{gather}
	
	Finally, substituting the results for the first term \eqref{eq:prop3125} and second term \eqref{eq:prop312} into \eqref{eq:prop33Ts1}, we have the variance bound
	\begin{gather} \label{eq:prop3313}
		\begin{split}
			& {\rm var} \left[ \bbu_s \right] \le 2 \alpha M K C_g^2 \| \bbx \|_2^2 p(1-p) + \ccalO(p^2 (1-p)^2)
		\end{split}
	\end{gather}
	completing the proof.
\end{proof}


\section{Proof of Theorem 1} \label{pr:theorem1}

In the proof, we need the following lemma with the proof in the supplementary material that shows the bound on the filter output.

\begin{lemma} \label{lemmma1}
	Consider the graph filter $\bbH(\bbS)$ [cf. \eqref{eq:sgf2} for $p = 1$] with coefficients $\{h_k\}_{k=0}^K$ and let $\bbS$ be the graph shift operator. Let the frequency response \eqref{eq:frere1} satisfy Assumption \ref{as:3} with constant $C_U$. Then, the norm of the graph filter is upper bounded as
	\begin{equation} \label{eq:lemmmamain}
		\begin{split}
			\| \bbH(\bbS) \|_2 \le C_U.
		\end{split}
	\end{equation}
\end{lemma}

\begin{proof}[Proof of Theorem \ref{theorem1}]
	From the SGNN definition in \eqref{eq.sgnn} and Assumption \ref{assumption3}, the variance can be upper bounded as
	\begin{equation} \label{eq:thm41}
		\begin{split}
			& \!{\rm var} \!\left[ \bbPhi(\bbx;\bbS_{P:1},\!\ccalH) \right] \!\!=\! {\rm var}\!\! \left[\! \sigma\!\!\left(\!\sum_{f\!=\!1}^{F}\! \bbu^f_{L-1}\!\!\right)\! \!\right] \!\!\!\le\!  {\rm var}\! \!\left[ \!\sum_{f\!=\!1}^{F}\! \bbu^f_{L\!-\!1}\!\! \right]\!
		\end{split}
	\end{equation}
	where ${\rm var}[\cdot]$ is defined in \eqref{eq:var0}. By exploiting the relation between the trace of covariance matrix and the variance, we can rewrite \eqref{eq:thm41} as
	\begin{align} \label{eq:thm42}
		{\rm var}\! \left[ \!\sum_{f=1}^{F}\! \bbu^f_{L-1}\! \right] &= \!\tr\! \left( \mathbb{E}\! \left[\!\big(\! \sum_{f=1}^{F}\! \bbu^f_{L-1}\!\big) \big(\! \sum_{f=1}^{F}\! \bbu^f_{L-1}\!\big)^\top\! \right]\! \right. \nonumber\\
		& \left. - \mathbb{E}\!\!\left[\!\sum_{f=1}^{F}\! \bbu^f_{L-1}\! \right]\!\! \mathbb{E}\!\!\left[\!\sum_{f=1}^{F}\! \bbu^f_{L-1}\! \right]^\top\! \right)\!. 
	\end{align}
	Denote $\bbu_{L-1}^f = \bbH_{L}^{f} \bbx^f_{L-1}$ and $\barbu_{L-1}^f=\barbH_L^f \bbx_{L-1}^f$ as concise notations of stochastic graph filter output $\bbH_{L}^f(\bbS_{K:0}) \bbx^f_{L-1}$ and expected graph filter output $\mathbb{E}\left[ \bbH_{L}^f(\bbS_{K:1}) \right]\bbx^f_{L-1}=\bbH_{L}^f(\barbS) \bbx^f_{L-1}$. By expanding \eqref{eq:thm42}, we get
	\begin{align} \label{eq:thm43}
		& {\rm var}\! \left[ \!\sum_{f=1}^{F}\! \bbu^f_{L-1}\! \right] \!= \!\sum_{f\!=\!1}^{F}\! \sum_{g\!=\!1}^{F}\tr\! \left(\!  \mathbb{E} \!\left[ \!\bbH_{L}^{f} \bbx^f_{L-1}\! \left( \bbH_{L}^{g} \bbx^{g}_{L-1} \right)^\top \!\right.  \right] \nonumber\\
		& \quad \quad \quad \quad \quad \quad \left.  - \mathbb{E}\left[ \bbH_{L}^{f} \bbx^f_{L-1} \right] \mathbb{E} \left[ \bbH_{L}^{g} \bbx^{g}_{L-1} \right]^\top \right).
	\end{align}
	By adding and subtracting $\barbH_{L}^{f} \bbx^f_{L-1} \left( \barbH_{L}^{g} \bbx^{g}_{L-1}\right)^\top$ inside the first expectation, \eqref{eq:thm43} becomes
	\begin{align} \label{eq:thm44}
		& \! \sum_{f\!=\!1}^{F}\! \sum_{g\!=\!1}^{F}\!\tr\! \left(\! \mathbb{E} \!\left[ \!\bbH_{L}^{f}\! \bbx^f_{L-1}\! \left( \bbH_{L}^{g}\! \bbx^{g}_{L-1} \right)^\top \!-\!\barbH_{L}^{f}\! \bbx^f_{L-1}\! \left( \barbH_{L}^{g}\! \bbx^{g}_{L-1}\right)^\top \!\right. \right] \nonumber\\
		&  \!\!+\!\mathbb{E}\!\left[ \barbH_{L}^{f}\! \bbx^f_{L\!-\!1}\! \left( \barbH_{L}^{g}\! \bbx^{g}_{L\!-\!1}\right)^\top \right]\! \left.\!- \mathbb{E}\!\left[ \bbH_{L}^{f} \bbx^f_{L\!-\!1} \!\right]\! \mathbb{E}\! \left[ \bbH_{L}^{g} \bbx^{g}_{L\!-\!1} \right]^\top\! \right)\!.\!
	\end{align}
	Expression \eqref{eq:thm44} is composed of two group of terms shown there in the two separate lines. 
	
	\textbf{First term.} For the first term, when $f \ne g$ such that filters $\bbH_{L}^{f}(\bbS_{K:0})$ and $\bbH_{L}^{g}(\bbS_{K:0})$ are independent, we have
	\begin{equation} \label{eq:thm455}
		\begin{split}
			& \!\tr\! \left(\! \mathbb{E} \!\left[ \!\bbH_{L}^{f}\! \bbx^f_{L\!-\!1}\! \left( \bbH_{L}^{g}\! \bbx^{g}_{L\!-\!1} \right)^\top \!-\!\barbH_{L}^{f}\! \bbx^f_{L-1}\! \left( \barbH_{L}^{g} \bbx^{g}_{L\!-\!1}\right)^\top \right] \right)\!=\! 0.
		\end{split}
	\end{equation}
	We then use \eqref{eq:thm455} to derive the upper bound
	\begin{align} \label{eq:thm45}
		& \!\sum_{f\!=\!1}^{F}\! \sum_{g\!=\!1}^{F}\!\tr\! \left(\! \mathbb{E} \!\left[ \!\bbH_{L}^{f}\! \bbx^f_{L-1}\! \left(\! \bbH_{L}^{g}\! \bbx^{g}_{L\!-\!1} \!\right)^\top \!-\!\barbH_{L}^{f}\! \bbx^f_{L-1}\! \left( \barbH_{L}^{g}\! \bbx^{g}_{L-1}\right)^\top \right]\! \right)\!\nonumber\\
		& = \!\sum_{f\!=\!1}^{F}\!\tr\! \left( \mathbb{E} \!\left[ \!\bbH_{L}^{f}\! \bbx^f_{L\!-\!1}\! \left(\! \bbH_{L}^{f}\! \bbx^{f}_{L\!-\!1} \!\right)^\top \!-\!\barbH_{L}^{f}\! \bbx^f_{L\!-\!1}\! \left( \!\barbH_{L}^{f}\! \bbx^{f}_{L\!-\!1}\!\right)^\top \!\right]\! \right)\! \nonumber\\
		& \le \Delta \sum_{f=1}^{F}\! \mathbb{E}\left[ \| \bbx^f_{L\!-\!1} \|_2^2 \right] p(1\!-\!p) \!+\! \mathcal{O}(p^2 (1\!-\!p)^2)
	\end{align}
	where $\Delta = 2\alpha M\! K\! C_g^2$ and the last inequality holds from Proposition \ref{proposition1}. For the norm of $\bbx_{L-1}^f$, we observe that
	\begin{equation}\label{eq:thm456}
		\begin{split}
			\mathbb{E}\!\!\left[ \| \bbx^f_{L\!-\!1} \|_2^2 \right] \!=\! \mathbb{E}\!\!\left[ \Big\| \sigma \!\!\left( \sum_{g=1}^F \bbu_{L\!-\!2}^{fg} \right) \!\Big\|_2^2 \right] \!\le\!C_\sigma^2 F\! \sum_{g=1}^F\! \mathbb{E}\!\left[ \left\| \bbu_{L\!-\!2}^{fg} \right\|_2^2 \right]
		\end{split}
	\end{equation}
	where Assumption \ref{assumption3} and the triangle inequality are used in the last inequality. By further representing $\big\| \bbu_{L\!-\!2}^{fg} \big\|_2^2$ with the trace $\tr (\bbu_{L-2}^{fg} (\bbu_{L-2}^{fg})^\top)$ and expanding the latter as in \eqref{eq:prop33Ts1}, we have
	\begin{align}
		\label{eq:thm458} &\!\!\mathbb{E}\!\left[\!\bbu_{L-2}^{fg} (\bbu_{L-2}^{fg})^\top\!\right]\! \!=\!\!\! \sum_{k,\ell=0}^K \!\!h_{k(L-2)}^{fg}h_{\ell(L-2)}^{fg}{\barbS}^{k}\mathbb{E}\!\left[\bbX_{L-2}^g\right] {\barbS}^{\ell}\\
		&\!\!\!\!\!+\!\!\!\!\!\sum_{k,\ell=0}^K\!\!\!\! h_{k(L\!-\!2)}^{fg}\!h_{\ell(L\!-\!2)}^{fg}\!\!\! \left(\!\! \mathbb{E}\!\!\left[\! \sum_{r=\!1}^{\lfloor\! k\ell\! \rfloor}\!\! \barbS^{k\!-\!r}\!\bbE_r\! \barbS^{r\!-\!1}\!\bbX_{L\!-\!2}^g \barbS^{r\!-\!1}\!\bbE_r\! \barbS^{\ell\!-\!r}\!\!\right]\!\! \!+\!\! \mathbb{E}\!\left[ \bbC_{k\ell}\right]\!\!\right) \nonumber 
	\end{align}
	with $\bbX_{L-2}^g = \bbx_{L-2}^g (\!\bbx_{L-2}^g)^\top$. For the first term, by using the cyclic property of trace and the inequality \eqref{eq:prop3ineq0}, we can bound it as
	\begin{align} \label{eq:thm459}
		&\tr \left( \sum_{k,\ell=0}^K\! h_{k(L\!-\!2)}^{fg}h_{\ell(L\!-\!2)}^{fg}{\barbS}^{k}\mathbb{E}\!\left[\bbX_{L-2}^g\right] {\barbS}^{\ell}\right)\nonumber\\
		& \le \| \sum_{k,\ell=0}^K \!h_{k(L-2)}^{fg}h_{\ell(L-2)}^{fg} \bar{\bbS}^{k+\ell} \|_2 \tr \left( \mathbb{E}\!\left[\bbX_{L-2}^g\right] \right)  \\
		& = \| \barbH_{L-2}^{fg} \barbH_{L-2}^{fg} \|_2 \tr \left( \mathbb{E}\!\left[\bbX_{L-2}^g \right]\right) \le C_U^2 \mathbb{E}\left[\|\bbx_{L-2}^{g} \|^2_2\right]\nonumber
	\end{align}
	where in the last inequality we used Lemma \ref{lemmma1} and $\tr \left( \mathbb{E}\left[\bbX_{L-2}^g \right]\right) = \mathbb{E}\left[\|\bbx_{L-2}^{g} \|^2_2\right]$. For the second term and the third term, we use the result \eqref{eq:prop312} and \eqref{eq:prop3125} to write
	\begin{align}
		\label{eq:thm4510} &\sum_{k,\ell\!=\!0}^K\!\!\! h_{k(L\!-\!2)}^{fg}\!h_{\ell(L\!-\!2)}^{fg} \!\tr\!\! \left(\!\! \mathbb{E}\!\!\left[\! \sum_{r\!=\!1}^{\lfloor\! k\ell\! \rfloor}\! \barbS^{k\!-\!r}\bbE_r\! \barbS^{r\!-\!1}\!\bbX_{L\!-\!2}^g \barbS^{r\!-\!1}\bbE_r\! \barbS^{\ell\!-\!r}\!\!+\!\! \bbC_{k\ell}\!\right]\!\!\right)\nonumber\\
		& \le \ccalO(p(1-p)).
	\end{align}
	By substituting \eqref{eq:thm459} and \eqref{eq:thm4510} into \eqref{eq:thm458} and the latter into \eqref{eq:thm456}, we get
	\begin{equation}\label{eq:thm4511}
		\begin{split}
			\mathbb{E}\!\!\left[ \| \bbx^f_{L\!-\!1} \|_2^2 \right] \!\le\!  C_\sigma^2 C_U^2 F \sum_{g=1}^F \mathbb{E}\!\!\left[ \left\| \bbx_{L\!-\!2}^{g} \right\|_2^2 \right] \!+\! \ccalO(p(1-p)).
		\end{split}
	\end{equation}
	Solving recursion \eqref{eq:thm4511} with the initial condition $\| \bbx_0^1 \|_2^2 = \| \bbx \|_2^2$ yields
	\begin{equation}\label{eq:thm4512}
		\begin{split}
			& \mathbb{E}[ \| \bbx_{L-1}^f \|_2^2]\! \le\!  C_\sigma^{2L-2}C_U^{2L\!-2}F^{2L\!-4}  \| \bbx \|_2^2 \!+\! \ccalO(p(1-p)).
		\end{split}
	\end{equation}
	By substituting \eqref{eq:thm4512} into \eqref{eq:thm45}, we bound \eqref{eq:thm45} by 
	\begin{equation} \label{eq:thm46}
		\begin{split}
			\Delta F^{2L-3} C_\sigma^{2L-2}C_U^{2L-2} \| \bbx \|_2^2 p(1-p)  + \mathcal{O}(p^2(1-p)^2).
		\end{split}
	\end{equation}
	
	\textbf{Second term.} For the second term in \eqref{eq:thm44}, we have\!\! \footnote{Likewise \eqref{eq:var0}, we use the notation ${\rm cov}[\bbx, \bby] = \sum_{i=1}^N {\rm cov}[[\bbx]_i, [\bby]_i]$ for any two random vectors $\bbx$ and $\bby$.}
	\begin{align} \label{eq:thm47}
		& \! \sum_{f\!=\!1}^{F}\! \sum_{g\!=\!1}^{F}\! \tr\!\left(\! \mathbb{E} \! \left[ \!\barbH_{L}^{f}\!\bbx^f_{L\!-\!1}\! \left( \!\barbH_{L}^{g}\!\bbx^{g}_{L\!-\!1}\!\right)^\top \! \right] \! -\! \mathbb{E}\!\left[\! \barbH_{L}^{f}\!\bbx^f_{L\!-\!1}\! \right]\! \mathbb{E} \!\left[\! \barbH_{L}^{g}\!\bbx^{g}_{L\!-\!1}\! \right]^\top \right) \nonumber\\
		& \!= \! \sum_{f=1}^{F} \! {\rm var} \! \left[ \!\barbH_{L}^{f}\!\bbx^f_{L-1} \!\right] \!+\! \sum_{f \ne g}^{F} \! {\rm cov} \! \left[ \! \barbH_{L}^{f}\!\bbx^f_{L-1}, \barbH_{L}^{g}\!\bbx^{g}_{L\!-\!1} \right].
	\end{align}
	By using the property of covariance \cite{Chandler1987}
	\begin{equation} \label{eq:thm48}
		\begin{split}
			{\rm cov}[x,y] \le \sqrt{{\rm var}[x]} \sqrt{{\rm var}[y]} \le \frac{{\rm var}[x]+{\rm var}[y]}{2} 
		\end{split}
	\end{equation}
	for two random variables $x$ and $y$, we can bound \eqref{eq:thm47} by
	\begin{align} \label{eq:thm485}
		&F \sum_{f=1}^{F} \! {\rm var} \! \left[ \!\barbH_{L}^{f}\!\bbx^f_{L-1} \!\right] =\! F\sum_{f=1}^{F} \tr\left( \barbH_{L}^{f} \barbH_{L}^{f}\bbSigma_{L-1}^f \right)
	\end{align}
	where $\bbSigma_{L-1}^f = \mathbb{E} \! \left[ \left( \!\bbx^f_{L\!-\!1}\! -\! \mathbb{E}\!\left[\! \bbx^f_{L\!-\!1}\! \right] \right) \left( \!\bbx^f_{L\!-\!1}\! -\! \mathbb{E}\!\left[\! \bbx^f_{L\!-\!1}\! \right] \right)^\top \right]$ is positive semi-definite matrix for all $f=1,\ldots,F$ and $\barbH_{L}^{f}\! \barbH_{L}^{f}$ is square. We then refer to inequality \eqref{eq:prop3ineq0} to bound \eqref{eq:thm485} as
	\begin{align} \label{eq:thm49}
		&\| \barbH_{L}^{g} \|_2^2 F \!\sum_{f\!=\!1}^{F}\! \tr\!\left( \bbSigma_{L-1}^f  \right)  \le C_U^2 F \sum_{g=1}^{F} {\rm var} \left[ \bbx^f_{L-1} \right]
	\end{align}
	where the last inequality is due to Lemma \ref{lemmma1}. By substituting bounds \eqref{eq:thm46} and \eqref{eq:thm49} into \eqref{eq:thm44} and then altogether into \eqref{eq:thm42}, we observe a recursion where the variance of $\ell$th layer output depends on the variance of $(\ell-1)$th layer output as well as the bound in \eqref{eq:thm46}. Therefore, we have
	\begin{align} \label{eq:thm410}
		&{\rm var}[\bbx_L]\!:=\!{\rm var} \left[ \bbPhi(\bbx;\bbS_{P:1},\!\ccalH) \right] \! \le\! C_U^2 F \!\sum_{f=1}^{F}\! {\rm var} \left[ \bbx^f_{L-1} \right]  \\
		& + \Delta F^{2L-3} C_\sigma^{2L-2}C_U^{2L-2} \| \bbx \|_2^2 p(1-p)  + \mathcal{O}(p^2(1-p)^2).\nonumber
	\end{align}
	
	\textbf{Unrolling.} By expanding this recursion until the input layer, we have
	\begin{align} \label{eq:thm411}
		& {\rm var}\!\! \left[ \bbPhi(\bbx;\bbS_{P:1},\!\ccalH) \right] \!\!\le\!\! \Delta\!\! \sum_{\ell\!=\!2}^L \!\!F^{2L\!-\!3} C_\sigma^{2\ell\!-\!2} C_U^{2L\!-\!2} \| \bbx \|_2^2 p(\!1\!-\!p\!) \\
		& \!\!+\!\! F^{2L\!-\!4} C_U^{2L\!-\!2}\!\! \left(\! \sum_{f\!=\!1}^{F} \!\!{\rm var}\! \left[ \bbx^f_1 \right] \!\!+ \!\!\!\sum_{f\ne g}^{F} \!\!{\rm cov}\! \left[ \!\bbx^f_1, \bbx^g_1 \!\right]\!\! \right)\! \!\!+\! \mathcal{O}\left(p^2(1\!-\!p)^2\right)\nonumber
	\end{align}
	where the first term in \eqref{eq:thm411} is the accumulated sum of all second terms in \eqref{eq:thm410} during the recursion. Since $\bbx_1^f = \bbH_{1}^f \bbx$ and $\bbx_1^g = \bbH_{1}^g \bbx$ and $\bbH_{1}^f$ and $\bbH_{1}^g$ are independent if $f \ne g$, we have $\sum_{f\ne g}^{F} \!{\rm cov}\!\! \left[ \bbx^f_1, \bbx^g_1 \right]=0$. Therefore, \eqref{eq:thm411} becomes
	\begin{align} \label{eq:thm412}
		& {\rm var} \left[ \bbPhi(\bbx;\bbS_{P:1},\!\ccalH) \right] \le \ccalO(p^2(1\!-\!p)^2) \\ 
		& ~~~~~~~+  2\alpha M \sum_{\ell=1}^L F^{2L-3} C_\sigma^{2\ell-2} C_U^{2L-2}K C_g^2 \| \bbx \|_2^2 p(1-p) \nonumber.
	\end{align}
	This completes the proof.
\end{proof}


\section{Proof of Theorem 2} \label{pr:theorem2}

\begin{proof}
	From the Taylor expansion of $\mathbb{E} \left[ \bar{C}(\ccalH_{t+1}) \right]$ at $\ccalH_t$, we have
	\begin{align} \label{eq:prthm11}
		&\mathbb{E}\!\left[ \bar{C}(\ccalH_{t\!+\!1})\right] \!\!=\! \mathbb{E}\!\left[ \bar{C}(\ccalH_{t}) \right. \!\!+\!\! \nabla_\ccalH \bar{C}(\ccalH_{t})^\top \!\left(\ccalH_{t+1}-\ccalH_t\right) \\
		&~~~~~~~~~~~~~~+\!\frac{1}{2}\left( \ccalH_{t\!+\!1}\!-\!\ccalH_t\right)^\top \nabla^2_\ccalH \bar{C} (\tilde{\ccalH}_{t}) \left( \ccalH_{t\!+\!1}-\ccalH_t\right) \big]\! \nonumber
	\end{align}
	where $\tilde{\ccalH}_{t}$ is in the line segment joining $\ccalH_{t+1}$ and $\ccalH_{t}$, and we substituted $\tilde{\ccalH}_{t}$ in the Hessian term $\nabla^2_\ccalH \bar{C} (\tilde{\ccalH}_{t})$ of \eqref{eq:prthm11} for the series truncation. Since for any vector $\bba$ and matrix $\bbA$, we have the inequality $\bba^\top \bbA \bba \le \lambda_{\max}(\bbA) \| \bbx \|_2^2$ where $\lambda_{\max}(\bbA)$ is the largest eigenvalue of $\bbA$. With this result and the Lipschitz continuity $\| \nabla^2_\ccalH \bar{C} (\tilde{\ccalH}_{t}) \|_2 \le C_L$ in Assumption \ref{as:1}, we get
	\begin{align} \label{eq:prthm12}
		&\mathbb{E}\!\!\left[ \bar{C}(\!\ccalH_{t+1}\!)\right] \\
		&\le \mathbb{E}\big[ \bar{C}(\!\ccalH_{t}\!)\! +\! \nabla_\ccalH \bar{C}(\!\ccalH_{t}\!)^\top\! \!\left(\! \ccalH_{t\!+\!1}\!-\!\ccalH_t\!\right) \!+\! \frac{C_L}{2}\! \| \ccalH_{t\!+\!1}\!-\!\ccalH_t \|_2^2\big]\!. \nonumber
	\end{align}
	By substituting the SGNN update rule $\ccalH_{t+1} = \ccalH_t - \alpha_t \nabla_\ccalH C(\bbS_{P:1}, \ccalH_t)$ with random cost realization, we get
	\begin{align} \label{eq:prthm125}
		\mathbb{E}\!\left[ \bar{C}(\ccalH_{t\!+\!1})\right]\! &\le \!\mathbb{E}\!\big[ \bar{C}(\ccalH_{t}) + \!\frac{\alpha_t^2 C_L}{2}\! \| \nabla_\ccalH C(\bbS_{P:1}, \ccalH_t) \|_2^2 \nonumber\\
		&- \alpha_t \nabla_\ccalH \bar{C}(\ccalH_t)^\top \nabla_\ccalH C(\bbS_{P:1}, \ccalH_t) \big].   
	\end{align}
	Exploiting the linearity of expectation and the identity $\mathbb{E}[\nabla_\ccalH C(\bbS_{P:1},\! \ccalH_t)]=\mathbb{E}[\nabla_\ccalH \bar{C}(\ccalH_t)]$, we can write \eqref{eq:prthm125} as
	\begin{align} \label{eq:prthm13}
		\mathbb{E}\left[ \bar{C}( \ccalH_{t+1})\right] &\le \mathbb{E}\big[ \bar{C}( \ccalH_t) ] \!-\! \alpha_t \mathbb{E}\big[\| \nabla_\ccalH \bar{C}(\ccalH_t) \|_2^2\big] \\
		&\!+\! \frac{\alpha_t^2 C_L}{2} \mathbb{E} [\| \nabla_\ccalH C(\bbS_{P:1}, \ccalH_t) \|_2^2\big].   \nonumber  
	\end{align}
	Subsequently, using the gradient bound in Assumption \ref{as:2}, we can upper bound the third term of \eqref{eq:prthm13} as
	\begin{equation} \label{eq:prthm14}
		\begin{split}
			\frac{\alpha_t^2 C_L}{2} \mathbb{E} [\| \nabla_\ccalH C(\bbS_{P:1},\! \ccalH_t) \|_2^2] \le \frac{\alpha_t^2 C_L C_B^2}{2}.
		\end{split}
	\end{equation}
	Since $\| \nabla_\ccalH \bar{C}(\ccalH_t) \|^2_2 \le \epsilon$ is the convergence criterion for non-convex problems, we focus on the second term in \eqref{eq:prthm13}. We first move it to the left side, and then move $\mathbb{E}\left[ \bar{C}(\ccalH_{t+1})\right]$ on the right side, and finally define the difference between two successive costs as
	\begin{align} \label{eq:prthm145}
		\Delta \bar{C}_{t:t+1} =\! \bar{C}( \ccalH_t) \!-\! \bar{C}(\ccalH_{t+1}).
	\end{align}
	With these arithmetic steps, we can rewrite \eqref{eq:prthm13} as
	\begin{equation} \label{eq:prthm15}
		\begin{split}
			\mathbb{E}[\| \nabla_\ccalH \bar{C}(\ccalH_t) \|_2^2] \le \frac{1}{\alpha_t} \mathbb{E}[ \Delta \bar{C}_{t:t+1} ]+ \frac{\alpha_t C_L C_B^2}{2}.
		\end{split}
	\end{equation}
	By considering a constant step $\alpha_t = \alpha$ and summing up all terms in \eqref{eq:prthm15} -- recall \eqref{eq:prthm15} should hold for all $t = 0, \cdots, T-1$ -- we get
	\begin{equation}\label{eq:prthm155}
		\begin{split}
			\sum_{t=0}^{T-1} \!\mathbb{E}[\| \nabla_\ccalH \bar{C}(\ccalH_t) \|^2] \le \frac{1}{\alpha} \mathbb{E}[ \Delta \bar{C}_{0:T} ]\!+\! \frac{ \alpha T C_L C_B^2}{2}.
		\end{split}
	\end{equation}
	For the optimal tensor $\ccalH^*$, we have $\bar{C} (\ccalH^*)\! \le\! \bar{C} ( \ccalH_{T})$. By substituting this result into $\Delta \bar{C}_{0:T}$, we bound \eqref{eq:prthm155} as
	\begin{align} \label{eq:prthm16}
		&\min_t \mathbb{E}[\| \nabla_\ccalH \bar{C}(\ccalH_t) \|^2] \le \frac{1}{T}\! \sum_{t=0}^{T-1}\! \mathbb{E}[\| \nabla_\ccalH \bar{C}(\ccalH_t) \|^2] \nonumber \\
		& \le \frac{1}{T \alpha} \left( \bar{C} (\ccalH_0) - \bar{C} (\ccalH^*) \right)+ \frac{ \alpha C_L C_B^2}{2}.
	\end{align}
	By further setting the constant step-size as
	\begin{equation} \label{eq:prthm165}
		\begin{split}
			\alpha \!=\! \sqrt{\frac{2\left( \bar{C} (\ccalH_0) - \bar{C} (\ccalH^*)\right)}{T C_L C_B^2 }}
		\end{split}
	\end{equation}
	and substituting it into \eqref{eq:prthm16}, we have
	\begin{equation} \label{eq:prthm17}
		\begin{split}
			&\min_t \mathbb{E}[\| \nabla_\ccalH \bar{C}(\ccalH_t) \|^2] \le \frac{C}{\sqrt{T}}
		\end{split}
	\end{equation}
	with constant $C\!=\! \sqrt{2\!\left( \bar{C} (\ccalH_0) - \bar{C} (\ccalH^*)\right)\!C_L}C_B$. Therefore, the bound decreases to zero with the rate $\mathcal{O}(1/\sqrt{T})$. This completes the proof.
\end{proof}


\bibliographystyle{IEEEtran}
\bibliography{myIEEEabrv,biblioSGNN}

\begin{table*}[!]\centering
	{\LARGE Supplementary Material for}:~
	{\LARGE  Stochastic Graph Neural Networks} \\
	{\large Zhan Gao$^{\dagger}$, Elvin Isufi$^{\ddagger }$ and Alejandro Ribeiro$^{\dagger}$ }
\end{table*}

\newpage

\appendices 

\section{Proof of Lemmas} \label{pr:lemmass}

\begin{proof}[Proof of Lemma \ref{lemma3}]
\textbf{Absolute value.} Denote with $\sigma(\cdot) = \vert \cdot \vert$ the nonlinearity, $x$ the random variable and $y = \sigma(x)$ the output. Since $y^2 = |x|^2 = x^2$, we have
\begin{equation} 
	\begin{split}
		{\rm var}[y] = \mathbb{E}[y^2] - \mathbb{E}[y]^2 & = \int_{-\infty}^{\infty} x^2 p(x)dx -\mathbb{E}[y]^2 \\
	\end{split}
\end{equation}
where $p(x)$ is the probability density function of $x$. From Jensen's inequality $\mathbb{E}[x]^2 \le \mathbb{E}[\vert x \vert]^2 = \mathbb{E}[y]^2$ and ${\rm var}[y] \ge 0$, we get
\begin{equation} \
	\begin{split}
		{\rm var}[y] \!=\!\! \int_{-\!\infty}^{\infty}\!\!\! x^2 p(x)dx\! -\!\mathbb{E}[y]^2\!  \le\!\! \int_{-\!\infty}^{\infty}\!\!\!\! x^2 p(x)dx\! -\!\mathbb{E}[x]^2\!\!=\!{\rm var}[x].
	\end{split}
\end{equation}

\textbf{ReLU.} Now denote with $\sigma(\cdot)= \max(0,\cdot)$ the nonlinearity. Similarly, we have
\begin{equation} 
	\begin{split}
		{\rm var}[y] \!=\! \mathbb{E}\![y^2] \!-\! \mathbb{E}\![y]^2  \!&=\! \int_{-\infty}^{\infty}\! \max(x,0)^2 p(x)dx \!-\!\mathbb{E}[y]^2\! \\
		& = \int_{0}^{\infty} x^2 p(x)dx -\mathbb{E}[y]^2 \\
	\end{split}
\end{equation}
and ${\rm var}[x] = \int_{-\infty}^{\infty} x^2 p(x)dx -\mathbb{E}[x]^2$. Thus, the difference of ${\rm var}[x]$ and ${\rm var}[y]$ is given by
\begin{equation} \label{prlemma1125}
	\begin{split}
		&{\rm var}[x]-{\rm var}[y]  = \int_{-\infty}^{0} x^2 p(x)dx -\mathbb{E}[x]^2 +  \mathbb{E}[y]^2 \\
		& = \int_{-\infty}^{0} x^2 p(x)dx+ (\mathbb{E}[y]+\mathbb{E}[x])(\mathbb{E}[y]-\mathbb{E}[x]).
	\end{split}
\end{equation}
Defining the integrals $A = \int_{-\infty}^{0} |x| p(x)dx$ and $B = \int_{0}^{\infty} x p(x)dx$, we have
\begin{equation} \label{prlemma115}
	\begin{split}
		\mathbb{E}[x] = B-A,~ \mathbb{E}[y] = B.
	\end{split}
\end{equation}
By substituting \eqref{prlemma115} into \eqref{prlemma1125}, we get
\begin{equation} \label{prlemma11}
	\begin{split}
		{\rm var}[x]-{\rm var}[y] & = \int_{-\infty}^{0} x^2 p(x)dx - A^2 +2AB
	\end{split}
\end{equation}
with $2AB \ge 0$ by definition. Consider the term $\int_{-\infty}^{0} x^2 p(x)dx - A^2$ and by defining the variable $z=\min( 0,x )$, we have
\begin{align} \label{prlemma12}
	&\int_{-\infty}^{0} x^2 p(x)dx \!-\! A^2 \!=\! \int_{-\infty}^{0} x^2 p(x)dx \!-\! \left(\! \int_{-\infty}^{0} \!\!|x| p(x)dx \!\right)^2 \nonumber\\
	& = \int_{-\infty}^{0} z^2 p(z)dz - \left( \int_{-\infty}^{0} z p(z)dz \right)^2  = {\rm var}[z] \ge 0
\end{align}
where $z \in (-\infty,0]$. By using \eqref{prlemma12} in \eqref{prlemma11}, we prove ${\rm var}[x]-{\rm var}[y] \ge 0$ completing the proof.
\end{proof}

\begin{proof}[Proof of Lemma \ref{lemmma3}]
The RES($\ccalG, p$) model samples each edge independently with probability $p$ such that $\barbS = p \bbS$. 

\textbf{Adjacency.} For the adjacency matrix $\bbS = \bbA$, the $(i,j)$th entry of $\bbS$ is $[\bbS]_{ij} = s_{ij}$ where $s_{ii}=0~\forall~i$. The $(i,j)$th entry of the RES($\ccalG, p$) realization $\bbS_k$ can be represented as $[\bbS_k]_{ij} = \delta_{ij}s_{ij}$ where $\delta_{ij}$ is a Bernoulli variable that is one with probability $p$ and zero with probability $1-p$. By exploiting the matrix multiplication for $\bbS_k\bbS_k$ and $\barbS\barbS$, the $(i,j)$th entries of $\mathbb{E}[\bbS_k^2]$ and $\barbS^2$ are respectively given by
\begin{align} \label{prlem31}
	\left[\mathbb{E}[\bbS_k^2]\right]_{ij} \!=\! \sum_{n\!=\!1}^N \!s_{in}s_{nj}\mathbb{E}[\delta_{in}\delta_{nj}],~ [\barbS^2]_{ij} \!=\! \sum_{n\!=\!1}^N\! s_{in}s_{nj} p^2
\end{align}
where in the second equality we also substituted $\barbS = p\bbS$. The Bernoulli variables $\{ \delta_{ij} \}_{ij}$ are independent except for $\delta_{ij}=\delta_{ji}$ since $\bbS_k$ is symmetric. Thus, we get 
\begin{equation}\label{prlem315}
	\mathbb{E}[ \delta_{in}\delta_{nj}] =
	\begin{cases}
		p^2,  & \mbox{if }i \ne j, \\
		p, & \mbox{if }i=j.
	\end{cases}
\end{equation}
By substituting \eqref{prlem315} into \eqref{prlem31}, we have
\begin{equation}\label{prlem32}
	\left[\mathbb{E}[\bbS_k^2]\right]_{ij} =
	\begin{cases}
		\sum_{n=1}^N \!s_{in}s_{nj}p^2,  & \mbox{if }i \ne j, \\
		\sum_{n=1}^N \!s_{in}s_{nj}p, & \mbox{if }i=j.
	\end{cases}
\end{equation}
Since $\sum_{n\!=\!1}^N \!s_{in}s_{ni} = d_i$ is the degree of node $i$ and from \eqref{prlem31} and \eqref{prlem32}, we can write
\begin{equation}\label{prlem325}
	\mathbb{E}\left[ \bbS_k^2 \right] = \barbS^2 + p(1-p)\bbD.
\end{equation}

\textbf{Laplacian.} For the adjacency matrix $\bbS = \bbL$, we have $s_{ii} =-\sum_{n \ne i} \delta_{in}s_{in}~\forall~i=1,\ldots,N$. In this case, $\mathbb{E}[\bbS_k^2]$ is
\begin{equation}\label{prlem33}
	\left[\!\mathbb{E}\![\bbS_k^2]\!\right]_{ij} \!\!=\!\!
	\begin{cases}
		\!\sum_{n \ne i,j} \!s_{in}s_{nj}p^2 \!\!+\!\mathbb{E}[ \delta_{ij}s_{ii}s_{ij} \!\!+\! \delta_{ij}s_{ij}s_{jj}],\! & \!\text{if }i \!\ne\! j\!, \\
		\!\sum_{n \ne i} \!s_{in}s_{nj}p \!+\! \mathbb{E}[ s_{ii}^2],\! & \!\text{if }i\!=\!j\!.
	\end{cases}
\end{equation}
For $\barbS^2$: if $i \ne j$, we have
\begin{align} \label{prlem34}
	[\barbS^2]_{ij} \!=\! \sum_{n\ne i,j} \! s_{in}s_{nj} p^2 \!-\!\sum_{n \ne i} s_{in}s_{ij} p^2\!-\!\sum_{n \ne j} s_{ij}s_{jn} p^2;
\end{align}
if $i=j$, we have
\begin{align} \label{prlem35}
	[\barbS^2]_{ij} \!=\! \sum_{n\ne i} \! s_{in}s_{nj} p^2 \!+\!\sum_{n_1 \ne i}\sum_{n_2 \ne i} s_{in_1}s_{in_2} p^2.
\end{align}
Now consider the terms $\mathbb{E}[\delta_{ij}s_{ii}s_{ij}]$, $\mathbb{E}[\delta_{ij}s_{ij}s_{jj}]$ and $\mathbb{E}[ s_{ii}^2]$. By expanding $s_{ii}$ and $s_{jj}$ and from \eqref{prlem315}, we have
\begin{subequations}\label{prlem36}
	\begin{align} 
		\mathbb{E}[s_{ii}\delta_{ij}s_{ij}] =  -\sum_{n \ne i,j} s_{in}s_{ij} p^2 - s_{ij}^2p,\\
		\mathbb{E}[s_{ij}\delta_{ij}s_{jj}] =  -\sum_{n \ne i,j} s_{ij}s_{jn} p^2 - s_{ij}^2p,\\
		\mathbb{E}[ s_{ii}^2] =  \sum_{n_1 \ne i}\sum_{n_2 \ne n_1,n_2\ne i} s_{in_1}s_{in_2} p^2 + \sum_{n \ne i} s_{in}^2 p.
	\end{align}
\end{subequations}
Substitute \eqref{prlem36} into \eqref{prlem33} and by comparing \eqref{prlem33} with \eqref{prlem34} and \eqref{prlem35} and using the fact $s_{ij} = -1$ for $i \ne j$, we get
\begin{equation}\label{prlem37}
	\mathbb{E}\left[ \bbS_k^2 \right] = \barbS^2 + 2p(1-p)\bbS
\end{equation}
which proves the lemma.
\end{proof}

\begin{proof}[Proof of Lemma \ref{lemmma1}]
For an input signal $\bbx$, we can write $i$th entry of the filter output $\bbu = \bbH(\bbS)\bbx$ in the GFT domain $\hat{\bbu}$ as $\hat{u}_{i} = h(\lambda_i)\hat{x}_i$ where $\hat{x}_i$ is $i$th coefficient of $\bbx$. From the energy conservation and Assumption 1 $|h(\lambda_i)| \le C_U$ for $i=1,\ldots, N$, we have
\begin{equation} \label{prpro3}
	\begin{split}
		\| \bbu \|^2_2 = \| \hat{\bbu} \|_2^2 = \sum_{i=1}^N h(\lambda_i)^2 \hat{x}_i^2 \le C_U^2 \| \bbx \|^2_2
	\end{split}
\end{equation}
which completes the proof.
\end{proof}

\end{document}